%% file: homeo.tex
\documentclass[12pt]{article}
\usepackage{microtype}
\usepackage{tikz-cd}

\pdfoutput=1

\input{environments}
\input{macros}

\usepackage[pagebackref=true]{hyperref}

\author{Jean-Daniel Boissonnat, Ramsay Dyer, Arijit
  Ghosh\thanks{Arijit Ghosh is supported by Ramanujan Fellowship
    (No. SB/S2/RJN-064/2015). Part of this work was done when Arijit
    Ghosh was a Researcher at Max-Planck-Institute for Informatics,
    Germany supported by the IndoGerman Max Planck Center
    for Computer Science (IMPECS).},\\
  and Mathijs Wintraecken}

\title{Local criteria for triangulation of manifolds\thanks{This work
    has been partially funded by the European Research Council under
    the European Union's ERC Grant Agreement number 339025 GUDHI
    (Algorithmic Foundations of Geometric Understanding in Higher
    Dimensions).}}

\begin{document}

\maketitle

\begin{abstract}
  We present criteria for establishing a triangulation of a
  manifold. Given a manifold $M$, a simplicial complex $\mathcal{A}$,
  and a map $H$ from the underlying space of $\mathcal{A}$ to $M$, our
  criteria are presented in local coordinate charts for $M$, and
  ensure that $H$ is a homeomorphism. These criteria do not require a
  differentiable structure, or even an explicit metric on $M$. No
  Delaunay property of $\mathcal{A}$ is assumed.  The result provides
  a triangulation guarantee for algorithms that construct a simplicial
  complex by working in local coordinate patches.  Because the
  criteria are easily verified in such a setting, they are expected to
  be of general use.
\end{abstract}

\clearpage

\tableofcontents

\input{intro}
\input{local_crit}
\input{distortion}
\input{subman}

\appendix
%\renewcommand{\thethm}{\Alph{section}.\arabic{thm}}

\input{degree}
\input{tanvar}
\begingroup % macro redefinitions in strong_diff
\input{strong_diff}
\endgroup

\bibliographystyle{alpha}
\bibliography{triang}

\end{document}

%% file: environments.tex
\usepackage{microtype}

\usepackage{amsmath,amssymb,euscript}
\usepackage{graphicx}
\usepackage{color}
\usepackage{ntheorem}

\def\theenumi{\arabic{enumi}}
\def\labelenumi{\textup{(\theenumi)}}

\newenvironment{proof}[1][{}]{%\novbskip%
  \begin{trivlist}\item[]\textit{Proof #1}\quad}%
  {\hfill\hspace*{\fill}~$\square$\end{trivlist}}

\newcommand{\noproof}{\hspace*{\fill}~$\square$}

\newtheorem{thm}{Theorem}%[section]
\newtheorem{prop}[thm]{Proposition}

\newtheorem{lem}[thm]{Lemma}
\newtheorem{cor}[thm]{Corollary}

\theorembodyfont{\normalfont}
\newtheorem{notation}[thm]{Notation}
\newtheorem{de}[thm]{Definition}

\newtheorem{remark}[thm]{Remark}
\newtheorem{hyp}[thm]{Hypothesis}

\theoremstyle{empty}

%% file: macros.tex
\newcommand{\Lemref}[1]{Lemma~\ref{#1}}
\newcommand{\Thmref}[1]{Theorem~\ref{#1}}
\newcommand{\Propref}[1]{Proposition~\ref{#1}}
\newcommand{\Secref}[1]{Section~\ref{#1}}
\newcommand{\Eqnref}[1]{Equation~\eqref{#1}}
\newcommand{\Defref}[1]{Definition~\ref{#1}}
\newcommand{\Hypref}[1]{Hypothesis~\ref{#1}}

\newcommand{\Appref}[1]{Appendix~\ref{#1}}
\newcommand{\Remref}[1]{Remark~\ref{#1}}

\newcommand{\defn}[1]{\emph{#1}}

\newcommand{\R}{\mathbb{R}}

\newcommand{\E}{\mathbb{E}}

\newcommand{\norm}[1]{\left|#1\right|}
\newcommand{\onorm}[1]{\left\|#1\right\|}
\newcommand{\abs}[2][foo]
{\csname#1l\endcsname\lvert{#2}\csname#1r\endcsname\rvert}
\newcommand{\transp}[1]{{#1}^\mathsf{T}}%\intercal}%\mathsf{T}}
\DeclareMathOperator{\sgn}{sgn}

\DeclareMathOperator{\Id}{Id}

\newcommand{\bracketprod}[2]{\left\langle#1,#2\right\rangle}
\newcommand{\bcdot}{\boldsymbol{\cdot}}
\newcommand{\dotprod}[2]{#1\bcdot#2}
\newcommand{\ip}[2]{\bracketprod{#1}{#2}} % use this

\DeclareMathOperator{\len}{len}

\newcommand{\man}{M}
\newcommand{\injradM}{\iota_{\man}}
\newcommand{\ballM}[2]{\spaceball{\man}{#1}{#2}}

\newcommand{\bdry}{\partial}
\newcommand{\close}[1]{\overline{#1}} % topological closure
\newcommand{\cl}[1]{\close{#1}}

\newcommand{\acplx}{\mathcal{A}}
\newcommand{\ccplx}{\mathcal{C}}

\newcommand{\carrier}[1]{\lvert#1\rvert}
\newcommand{\St}{\mathrm{St}}
\newcommand{\str}[1]{\underline{\St}(#1)}
\DeclareMathOperator{\interior}{int}
\newcommand{\intr}[1]{\interior(#1)}
\newcommand{\ostr}[1]{\mathrm{st}(#1)}
\DeclareMathOperator{\rinterior}{relint}
\newcommand{\rint}[1]{\rinterior(#1)}

\newcommand{\lmap}{\widehat{\Phi}}

\newcommand{\pts}{\mathcal{P}}

\newcommand{\thickbnd}{t_0}

\newcommand{\splxs}{\sigma}
\newcommand{\gsplxs}{\boldsymbol{\sigma}}
\newcommand{\splxsE}{\gsplxs_{\E}}

\newcommand{\gsplxt}{\boldsymbol{\tau}}

\newcommand{\gsplxm}{\boldsymbol{\mu}}

\newcommand{\spaceball}[3]{B_{#1}(#2,#3)} % open ball on specified space
\newcommand{\cspaceball}[3]{\close{B}_{#1}(#2,#3)} % closed ball " " "

\newcommand{\rem}{\R^m}
\newcommand{\ballEm}[2]{\spaceball{\rem}{#1}{#2}} % open ball in R^m
\newcommand{\ballEN}[2]{\spaceball{\amb}{#1}{#2}} % open ball in R^N
\newcommand{\cballEN}[2]{\cspaceball{\amb}{#1}{#2}} % closed ball in R^N

\newcommand{\gdist}{d} % generic metric (distance function)

\newcommand{\gdistG}[1]{\gdist_{#1}} % metric on space to be specified
\newcommand{\distG}[3]{\gdist_{#1}(#2,#3)}

\newcommand{\distEm}[2]{\distG{\rem}{#1}{#2}}

\newcommand{\amb}{\R^N}
\newcommand{\distEN}[2]{\distG{\amb}{#1}{#2}}

\newcommand{\gdtoM}{\delta_{\man}} 
\newcommand{\dtoM}[1]{\gdtoM(#1)}  %{\distEN{#1}{\man}}

\newcommand{\distM}[2]{\distG{\man}{#1}{#2}}

\newcommand{\tccplx}{\widetilde{\ccplx}}

\newcommand{\p}{\hat{p}}
\newcommand{\strp}{\str{\p}}
\newcommand{\x}{\hat{x}}
\newcommand{\y}{\hat{y}}
\newcommand{\tV}{\widetilde{V}}

\newcommand{\npwr}[1]{\wp(#1)} %nonempty subsets
\newcommand{\conset}{\boldsymbol{\omega}}
\newcommand{\proj}[1]{\mathrm{pr}_{\!#1}}   %{\Pi_{#1}}
\newcommand{\projp}{\proj{T_p\man}}
\newcommand{\projM}{\proj{\man}}

\def\aff{\mathrm{aff}}
\def\ax{\mathop{\mathrm{ax}}\nolimits}

\def\claxM{\overline{\ax}(\man)}
\def\px{\check{x}}
\def\py{\check{y}}
\def\pz{\check{z}}
\def\rch{\mathop{\mathrm{rch}}\nolimits}
\def\rchM{\rch(M)}
\def\UM{U_{\!\man}}
\def\lfs{\mathop{\mathrm{lfs}}\nolimits}
\def\NM{N\!\man}

\def\tfrac#1#2{{\textstyle\frac#1#2}}
\def\rbnd{R_\mathrm{rch}}
\def\tx{\widetilde{x}}
\def\ty{\widetilde{y}}
\def\tz{\widetilde{z}}

\def\Lbnd{L_0}
\def\tbnd{t_0}
\def\sbnd{s_0}
\def\chnk{\frac{\Lbnd^2}{\tbnd^2\rbnd^2}}
\def\chink#1{\frac{#1\Lbnd^2}{\tbnd^2\rbnd^2}}
\def\htbnd{\hat{t}_0}
\def\hLbnd{\hat{L}_0}
\def\hsbnd{\hat{s}_0}
\def\sepbnd{\mu_0}
\def\cmpman{\man_{c}}

\def\hy{\widehat{y}}
\def\hpy{\widehat{\py}}
\def\id{\mathop{\mathrm{id}}\nolimits}

\def\cbp/{cbp}

\def\unpfrac#1#2{#1/(#2)}

%% file: intro.tex
\section{Introduction}

A \defn{triangulation} of a manifold $\man$ is a homeomorphism
$H\colon \carrier{\acplx} \to \man$, where $\acplx$ is a simplicial
complex, and $\carrier{\acplx}$ is its 
%\defn{carrier}, i.e., the
underlying topological space. If such a homeomorphism exists, we say
that $\acplx$ \defn{triangulates} $\man$.

The purpose of this paper is to present criteria which ensure that a
candidate map $H$ is indeed a homeomorphism.  This work is motivated
by earlier investigations into the problem of algorithmically
constructing a complex that triangulates a given manifold
\cite{boissonnat2014tancplx,boissonnat2017manmesh}. It complements and
is closely related to recent work that investigates a particular
natural example of such a map \cite{dyer2015riemsplx}.

In the motivating algorithmic  setting, we are given a compact manifold
$\man$, and a manifold simplicial complex $\acplx$ is constructed by
working locally in Euclidean coordinate charts.  Here we lay out criteria,
 based on local properties that arise naturally in the
construction of $\acplx$, that guarantee that $H$ is a homeomorphism.
These criteria, which are summarized in
\Thmref{thm:metric.triang}, are based on metric
properties of $H$ within ``compatible'' coordinate charts
(\Defref{def:compatible.atlases}). The Euclidean metric in the local
coordinate chart is central to the analysis, but no explicit metric on
$\carrier{\acplx}$ or $\man$ is involved, and no explicit assumption
of differentiability is required of $H$ or $\man$. However, our only
examples that meet the required local criteria are in the
differentiable setting. We do not know whether or not our criteria for
homeomorphism implicitly imply that $\man$ admits a differentiable
structure. They do imply that $\acplx$ is piecewise linear (admits an
atlas with piecewise linear transition functions).

\subsection*{Relation to other work}

The first demonstrations that differentiable manifolds \emph{can}
always be triangulated were constructive.  Cairns~\cite{cairns1934}
used coordinate charts to cover the manifold with embeddings of
patches of Euclidean triangulations. He showed that if the complexes
were sufficiently refined the embedding maps could be perturbed such
that they remain embeddings and the images of simplices coincide where
patches overlap. A global homeomorphic complex is obtained by
identifying simplices with the same image. The technique was later
refined and extended \cite{whitehead1940,munkres1968}, but it is not
easily adapted to provide triangulation guarantees for complexes
constructed by other algorithms.

An alternative approach was developed by Whitney~\cite{whitney1957}
using his result that a manifold can be embedded into Euclidean
space. A complex is constructed via a process involving the
intersection of the manifold with a fine Cartesian grid in the ambient
space, and it is shown that the \defn{closest-point projection map},
which takes a point in the complex to its unique closest point in the
manifold, is a homeomorphism. The argument is entwined with this
specific construction, and is not easily adapted to other settings.

More recently, Edelsbrunner and Shah~\cite{edelsbrunner1997rdt}
defined the restricted Delaunay complex of a subset $M$ of Euclidean
space as the nerve of the Voronoi diagram on $M$ when the ambient
Euclidean metric is used. They showed that if $M$ is a compact
manifold, then the restricted Delaunay complex is homeomorphic to $M$
when the Voronoi diagram satisfies the \defn{closed ball property
  (\cbp/)}: Voronoi faces are closed topological balls of the
appropriate dimension.

Using the \cbp/, Amenta and Bern~\cite{amenta1999vf}
demonstrated a specific sampling density that is sufficient to
guarantee that the restricted Delaunay complex triangulates the
surface. However, since the complex constructed by their
reconstruction algorithm cannot be guaranteed to be exactly the
restricted Delaunay complex, a new argument establishing homeomorphism
was developed, together with a simplified version of the algorithm
\cite{amenta2002}.

Although it was established in the context of restricted Delaunay
triangulations, the \cbp/ is an elegant topological result that
applies in more general contexts. For example, it has been used to
establish conditions for intrinsic Delaunay triangulations of surfaces
\cite{dyer2008sgp}, and Cheng et al.~\cite{cheng2005} have indicated
how it can be applied for establishing weighted restricted Delaunay
triangulations of smooth submanifolds of arbitrary dimension in
Euclidean space.

However, the \cbp/ is only applicable to Delaunay-like complexes that
can be realized as the nerve of some kind of Voronoi diagram on the
manifold. Thus, for example, it does not necessarily apply to the
tangential Delaunay complex constructed by Boissonnat and
Ghosh~\cite{boissonnat2014tancplx}. Secondly, even when a
Delaunay-like complex is being constructed, it can be 
%algorithmically
difficult to directly verify the properties of the associated Voronoi
structure; sampling criteria and conditions on the complex under
construction are desired, but may not be easy to obtain from the \cbp/. 
A third deficiency of the \cbp/ is that, although it can
establish that a complex $\acplx$ triangulates the manifold $\man$, it
does not provide a specific triangulation
$H:\carrier{\acplx} \to \man$. Such a correspondence allows us to
compare \emph{geometric} properties of $\carrier{\acplx}$ and~$\man$.

In \cite{boissonnat2014tancplx} Whitney's argument was adapted to
demonstrate that the closest-point projection maps the tangential
Delaunay complex homeomorphically onto the original manifold. The
argument is intricate, and like Whitney's, is tailored to the specific
complex under consideration. In contrast, the result of
\cite{amenta2002}, especially in the formulation presented by
Boissonnat and Oudot~\cite{boissonnat2005gm}, guarantees a
triangulation of a surface by any complex which satisfies a few easily
verifiable properties. However, the argument relies heavily on the
the codimension being 1.

If a set of vertices is contained within a sufficiently small
neighbourhood on a Riemannian manifold, barycentric coordinates can be
defined. So there is a natural map from a Euclidean simplex of the
appropriate dimension to the manifold, assuming a correspondence
between the vertices of the simplex and those on the manifold.  Thus
when a complex $\acplx$ is appropriately defined with vertices on a
Riemannian manifold $\man$, there is a natural \defn{barycentric
  coordinate map} $\carrier{\acplx} \to \man$.  In
\cite{dyer2015riemsplx}, conditions are presented which guarantee that
this map is a triangulation. Although this map is widely applicable,
the intrinsic criteria can be inconvenient, for example, in the
setting of Euclidean submanifold reconstruction, and furthermore the
closest-point projection map may be preferred for
triangulation in that setting. 

The argument in \cite{dyer2015riemsplx} is based on a general result
\cite[Proposition~16]{dyer2015riemsplx} for establishing that a given
map is a triangulation of a differentiable manifold. However, the
criteria include a bound on the differential of the map, which is not
easy to obtain. The analysis required to show that the closest-point
projection map meets this bound is formidable, and this motivated the
current alternate approach.  We have relaxed this constraint to a much
more easily verifiable bound on the metric distortion of the map when
viewed within a coordinate chart.

The sampling criteria for submanifolds imposed by our main result
applied to the closest-point projection map
(\Thmref{thm:triang.subman}) are the most relaxed that we are aware
of.  The result could be applied to improve the sampling guarantees
of previous works, e.g., \cite{cheng2005,boissonnat2014tancplx}.

In outline, the argument we develop here is the same as that of
\cite{amenta2002}, but extends the result to apply to abstract
manifolds of arbitrary dimension and submanifolds of $\R^N$ of
arbitrary codimension. We first show that the map $H$ is a local
homeomorphism, and thus a covering map, provided certain criteria are
met. Then injectivity is ensured when we can demonstrate that each
component of $\man$ contains a point $y$ such that $H^{-1}(y)$ is a
single point. A core technical lemma from Whitney~\cite[Appendix~II
Lemma~15a]{whitney1957} still lies at the heart of our argument.

\subsection*{Outline}
The demonstration is developed abstractly in \Secref{sec:homeo.criteria},
without explicitly defining the map $H$. We assume that it has already
been established that the restriction of $H$ to any Euclidean simplex
in $\carrier{\acplx}$ is an embedding. This is a nontrivial step that
needs to be resolved from the specific properties of a particular
choice of $H$. The criteria for local homeomorphism apply in a common
coordinate chart (for $\carrier{\acplx}$ and $\man$), and relate the
size and quality of the simplices with the metric distortion of $H$,
viewed in the coordinate domain. The requirement that leads to
injectivity is also expressed in a local coordinate chart; it
essentially demands that the images of vertices behave in a natural
and expected way.

In \Secref{sec:subman}, \Thmref{thm:metric.triang} is applied to the
specific case where $\man \subset \R^N$, and $H$ is the projection to
the closest-point on $\man$. This is a refinement of the argument
presented in \cite{boissonnat2014tancplx}, also correcting an error.

In \Appref{sec:exploit.strong.diff.bnds} the argument presented in
\cite{dyer2015riemsplx} 
%for the barycentric coordinate map 
is reviewed. This is very similar to the argument presented in
\Secref{sec:homeo.criteria}, but it exploits a detailed analysis of
the differential of the map. Although it depends on differentiability,
and more analysis, an advantage of this approach is that the size of
the simplices is restricted by a bound that is linear with respect to
the simplex quality, instead of quadratic.
% no bound on
% the relative diameters of the simplices is required (simplices can be
% arbitrarily small relative to the upper bound imposed on the
% diameters).  
The main reason for reviewing this argument is that as a result of an
error in \cite{dyer2015riemsplx}, the criteria in the announced
theorems there do not in fact ensure injectivity of the map. The
appendix illuminates the problem and corrected statements of those
results are presented in \Secref{sec:bary.map.erratum}.

%%% Local Variables:
%%% mode: latex
%%% TeX-master: "homeo"
%%% End:

%% file: local_crit.tex
\section{The homeomorphism criteria}
\label{sec:homeo.criteria}

We assume that $\acplx$ and $\man$ are both compact manifolds of
dimension $m$, without boundary, and we have a map
$H\colon \carrier{\acplx} \to \man$ that we wish to demonstrate is a
homeomorphism.  We first show that $H$ is a covering map, i.e., every
$y\in\man$ admits an open neighbourhood $U_y$ such that $H^{-1}(y)$ is
a disjoint union of open sets each of which is mapped homeomorphically
onto $U_y$ by $H$.  In our setting it is sufficient to establish that
$H$ is a local homeomorphism whose image touches all components of
$\man$: Brouwer's invariance of domain then ensures that $H$ is
surjective, and, since $\carrier{\acplx}$ is compact, has the covering
map property.

\begin{notation}[simplices and stars]
  \label{not:splx.star}
  In this section, a simplex $\gsplxs$ will always be a \defn{full
    simplex}: a closed Euclidean simplex, specified by a set of
  vertices together with all the points with nonnegative barycentric
  coordinates.  The \defn{relative interior} of $\gsplxs$ is the
  topological interior of $\gsplxs$ considered as a subspace of its
  affine hull, and is denoted by $\rint{\gsplxs}$.  If $\gsplxs$ is a
  simplex of $\acplx$, the subcomplex consisting of all simplices that
  have $\gsplxs$ as a face, together with the faces of these
  simplices, is called the \defn{star} of $\gsplxs$, denoted
  by~$\str{\gsplxs}$; the star of a vertex $p$ is $\str{p}$.

  We also sometimes use the \defn{open star} of a simplex $\gsplxs \in
  \ccplx$. This is the union of the relative interiors of the
  simplices in $\ccplx$ that have $\gsplxs$ as a face:
  $\ostr{\gsplxs} = \bigcup_{\gsplxt \supseteq
    \gsplxs}\rint{\gsplxt}$.  It is an open set in $\carrier{\ccplx}$,
  and it is open in $\R^m$ if $\gsplxs \not\in \bdry{\ccplx}$.
\end{notation}

\begin{notation}[topology]
  If $A \subseteq \R^m$, then the topological closure, interior, and
  boundary of $A$ are denoted respectively by $\cl{A}$,\, $\intr{A}$,
  and  $\bdry{A} = \cl{A} \setminus \intr{A}$. We denote by
  $\ballEm{c}{r}$ the open ball in $\R^m$ of radius $r$ and centre $c$.
\end{notation}

\begin{notation}[linear algebra]
  The Euclidean norm of $v\in\R^m$ is denoted by $\norm{v}$, and
  $\onorm{A}=\sup_{\norm{x}=1}\norm{Ax}$ denotes the operator norm of
  the linear operator~$A$.
\end{notation}

We will work in local coordinate charts. 
% Given a map $G\colon \carrier{\ccplx} \to \rem$, where $\ccplx$
% is a simplicial complex,
% the \defn{secant map} of $G$ is the piecewise linear
% approximation to $G$ that agrees with $G$ on the vertices of $\ccplx$,
% and maps $x \in \gsplxs \in \ccplx$ to the point with the same
% barycentric coordinates with respect to the images of the vertices.
To any given map $G\colon \carrier{\ccplx} \to \rem$, where $\ccplx$
is a simplicial complex, we associate a piecewise linear map $\widehat{G}$
that agrees with $G$ on the vertices of $\ccplx$, and maps
$x \in \gsplxs \in \ccplx$ to the point with the same barycentric
coordinates with respect to the images of the vertices. The map
$\widehat{G}$ is called the \defn{secant map} of $G$ with respect to
$\ccplx$.

The following definition provides the framework within which we will
work (see diagram~\eqref{diagram}).

\begin{de}[compatible atlases]
  \label{def:compatible.atlases}
  We say that $\carrier{\acplx}$ and $\man$ have \defn{compatible
    atlases} for $H\colon \carrier{\acplx} \to \man$ if: 
  \begin{enumerate}
  \item There is a coordinate atlas $\{(U_p, \phi_p)\}_{p\in\pts}$ for
    $\man$, where the index set $\pts$ is the set of vertices of
    $\acplx$ and each set $U_p$ is connected.
 % together with a compatible piecewise linear
 %  coordinate atlas for $\acplx$ in the following sense: 
  \item For each $p\in \pts$, there is a subcomplex $\tccplx_p$ of
    $\acplx$ that contains $\str{p}$ and 
    $H(\carrier{\tccplx_p}) \subset U_p$. Also, the secant map of
    $\Phi_p := \phi_p \circ H|_{\carrier{\tccplx_p}}$ defines a
    piecewise linear embedding of $\carrier{\tccplx_p}$ into
    $\rem$. We denote this secant map by $\lmap_p$. By definition,
    $\lmap_p$ preserves the barycentric coordinates within each
    simplex, and thus the collection
    $\{(\tccplx_p,\lmap_p)\}_{p\in \pts}$ provides a piecewise linear
    atlas for $\acplx$.
\end{enumerate}
\end{de}

\begin{remark}
  \label{rem:touch.all.components}
  The requirement in \Defref{def:compatible.atlases} that the local
  patches $U_p$ be connected implies that on each connected component
  $M'$ of $M$, there is a $p\in \pts$ such that $H(p)\in M'$.
\end{remark}

We let $\ccplx_p = \lmap_p(\tccplx_p)$, and we will work within
the compatible local coordinate charts.
Thus we are studying a map of the form
\begin{equation*}
  F_p \colon \carrier{\ccplx_p} \subset \rem \to \rem,
\end{equation*}
where $\ccplx_p$ is an $m$-manifold complex with boundary embedded in
$\rem$, and 
\begin{equation}
\label{eq:def.Fp}
F_p = \phi_p \circ H \circ \lmap_p^{-1},
\end{equation}
as shown in the following diagram:
% \begin{equation*} %good, modern syntax:
% \begin{tikzcd}
% \carrier{\acplx} \ar[rrr, "H"]&{} &{} &\man \\[-10pt]
%   &\carrier{\smash{\tccplx_p}} \arrow[d, "\lmap_p"']
% \ar[ul, hook'] \ar[r, "H|_{|\tccplx_p|}"] 
%     &U_p \arrow[d, "\phi_p"] \ar[ur, hook] &{}\\
%     &\carrier{\ccplx_p} \ar[ld, hook']
%     \arrow[r, "\overbrace{\scriptstyle \phi_p \circ H
%       \circ \lmap_p^{-1}}^{F_p}"']
%     &\phi_p(U_p) \ar[rd, hook]&{}\\[-5pt]
% \rem &{} &{} &\rem
% \end{tikzcd}
% \end{equation*}
\begin{equation} %old syntax for J-D:
\label{diagram}
\begin{tikzcd}
\carrier{\acplx} \ar{rrr}{H} \ar[hookleftarrow]{dr} 
& & &\man \\[-10pt]
  &\carrier{\smash{\tccplx_p}} \arrow{d}[swap]{\lmap_p}
 \ar{r}{H|_{|\tccplx_p|}}
    & U_p \arrow{d}{\phi_p} \ar[hook]{ur} &\\
    &\carrier{\ccplx_p}
    \arrow{r}[swap]{\overbrace{\scriptstyle \phi_p \circ H
      \circ \lmap_p^{-1}}^{F_p}}
    & \phi_p(U_p) \ar[hook]{rd}&\\[-5pt]
\rem \ar[hookleftarrow]{ur} & & &\rem
\end{tikzcd}
\end{equation}
We 
will focus on the map $F_p$, which can be considered as a local
realisation of $H|_{\carrier{\tccplx_p}}$. By construction, $F_p$
leaves the vertices of $\ccplx_p$ fixed: if $q \in \rem$ is a vertex
of $\ccplx_p$, then $F_p(q)=q$, since $\lmap_p$ coincides with
$\phi_p \circ H$ on vertices.

\begin{remark}
  The setting described here conforms to the
  paradigm laid out in the tangential complex work
  \cite{boissonnat2014tancplx}.  There one locally constructs a
  (weighted) Delaunay triangulation in the tangent space
  $T_p\man$. This gives us the local patch $\ccplx_p$ (the star of $p$
  in $T_p\man$). The vertices of the constructed complex actually lie
  on $\man$, and we recognise $\ccplx_p$ as the orthogonal projection
  $\lmap_p$ of the corresponding complex $\tccplx_p$ with
  vertices on $\man$.

  In this context, the homeomorphism $H$ that we are trying to
  establish is given by the closest-point projection map onto $\man$,
  restricted to $\carrier{\acplx}$. Now we are going to
  work in local coordinate charts, given at each vertex $p \in \pts$
  by the orthogonal projection $\phi_p$ of some neighbourhood of $p$,
  $U_p \subset \man$ into $T_p\man$. We recognise that $\lmap_p$
  really does coincide with the secant map of
  $\phi_p \circ H|_{|\tccplx_p|}$.
\end{remark}

\subsection{Local homeomorphism}

Our goal is to ensure that there is some open
$V_p \subset \carrier{\ccplx_p}$ such that $F_p|_{V_p}$ is an
embedding and that the sets $\tV_p = \lmap_p^{-1}(V_p)$ are
sufficiently large to cover $\carrier{\acplx}$. This will imply that
$H$ is a local homeomorphism. Indeed, if $V_p$ is embedded by $F_p$,
then $\tV_p$ is embedded by
$H|_{\tV_p} = \phi_p^{-1}\circ F_p \circ \lmap_p
|_{\tV_p}$, since $\phi_p$ and $\lmap_p$ are both embeddings.
Since $\carrier{\acplx}$ is compact, Brouwer's invariance of domain,
together with 
%the last assertion of \Defref{def:compatible.atlases}
\Remref{rem:touch.all.components},
implies that $H$ is surjective, and a covering map. It will only
remain to ensure that $H$ is also injective.

We assume that we are given (i.e., we can establish by
context-dependent means) a couple of properties of $F_p$. We
assume that it is \defn{simplexwise positive}, which means that 
it is continuous, and its restriction to any
$m$-simplex in $\strp$ is an orientation preserving
topological embedding. As discussed in \Appref{sec:degree.theory}, we
say that $F_p$ preserves the orientation of an $m$-simplex
$\gsplxs \subset \R^m$ if $F_p|_{\gsplxs}$ has degree 1 at any point
in the image of the interior of $\gsplxs$, i.e.,
$\deg(F_p,\intr{\gsplxs},y)=1$ for $y\in F_p(\intr{\gsplxs})$. (The
other assumption we make is that $F_p$, when restricted to an
$m$-simplex does not distort distances very much, as discussed below.)

The local homeomorphism demonstration is based on 
\Lemref{lem:splx.pos.embed} below,
which is a particular case of an observation made by Whitney
\cite[Appendix~II Lemma 15a]{whitney1957}. Whitney demonstrated a more
general result from elementary first principles. The proof we give
here is the same as Whitney's, except that we exploit elementary
degree theory, as discussed in \Appref{sec:degree.theory}, in order to
avoid the differentiability assumptions Whitney made. 

In the statement of the lemma, $\ccplx^{m-1}$ refers to the
\defn{$(m{-}1)$-skeleton} of $\ccplx$: the subcomplex consisting of
simplices of dimension less than or equal to $m-1$. When
$\carrier{\ccplx}$ is a manifold with boundary, as in the lemma, then
$\bdry\ccplx$ is the subcomplex containing all $(m{-}1)$-simplices
that are the face of a single $m$-simplex, together with the faces of
these simplices.

\begin{lem}[simplexwise positive embedding]
  \label{lem:splx.pos.embed}
  Assume $\ccplx$ is an oriented $m$-manifold finite simplicial
  complex with boundary embedded in $\R^m$.  Let
  $F \colon \carrier{\ccplx} \to \rem$ be simplexwise positive in
  $\ccplx$.  Suppose $V \subset \carrier{\ccplx}$ is a connected open
  set such that $F(V) \cap F(\carrier{\bdry{\ccplx}}) = \emptyset$. If
  there is a $y \in F(V) \setminus F(\carrier{\ccplx^{m-1}})$ such
  that $F^{-1}(y)$ is a single point, then the restriction of $F$ to
  $V$ is a topological embedding.
\end{lem}

\begin{proof}
  Notice that the topological boundary of
  $\carrier{\ccplx} \subset \R^m$ is equal to the carrier of the
  boundary complex (see, e.g., \cite[Lemmas~3.6,
  3.7]{boissonnat2013stab1}):
  \begin{equation*}
    \bdry{\carrier{\ccplx}} = \carrier{\bdry{\ccplx}}.
  \end{equation*}
  Let $\Omega = \carrier{\ccplx} \setminus \carrier{\bdry{\ccplx}}$.
  Since $F$ is simplexwise positive, and $F(V)$ lies within a
  connected component of $\R^m \setminus F(\bdry\Omega)$, the fact
  that $F^{-1}(y)$ is a single point implies that $F^{-1}(w)$ is a
  single point for any $w \in F(V) \setminus F(\carrier{C^{m-1}})$
  (\Lemref{lem:splxwise.pos.loc.cnst}).
  We need to show that $F$ is also injective on
  $V \cap \carrier{\ccplx^{m-1}}$. 

  We now show that $F(\ostr{\gsplxs})$ is open for any
  $\gsplxs \in \ccplx^{m-1} \setminus \bdry{\ccplx}$, where
  $\ostr{\gsplxs}$ is the open star of $\gsplxs$, defined in
  Notation~\ref{not:splx.star}.  Suppose $x \in \rint{\gsplxt}$ for
  some $\gsplxt \in \ccplx \setminus \bdry{\ccplx}$.  Since $F$ is
  injective when restricted to any simplex, we can find a sufficiently
  small open (in $\R^m$) neighbourhood $U$ of $F(x)$ such that
  $U \cap F(\bdry{\ostr{\gsplxt}}) = \emptyset$. Since the closure of
  the open star is equal to the carrier of our usual star:
  \[
    \cl{\ostr{\gsplxt}} = \carrier{\str{\gsplxt}},
  \]
  \Lemref{lem:splxwise.pos.loc.cnst} implies that every point in
  $U \setminus F(\carrier{\str{\gsplxt}^{m-1}})$ has the same number
  of points in its preimage. By the injectivity of $F$ restricted to
  $m$-simplices, this number must be greater than zero for points near
  $F(x)$.  It follows that $U \subseteq F(\ostr{\gsplxt})$.

  If $x \in \ostr{\gsplxs}$, then $x \in \rint{\gsplxt}$ for some
  $\gsplxt \in \ccplx \setminus \bdry{\ccplx}$ that has $\gsplxs$ as a
  face. Since $\ostr{\gsplxt} \subseteq \ostr{\gsplxs}$, we have
  $U \subseteq F(\ostr{\gsplxs})$, and we conclude that
  $F(\ostr{\gsplxs})$ is open.

  Now, to see that $F$ is injective on
  $\carrier{\ccplx^{m-1}} \cap V$, suppose to the contrary that
  $w,z \in \carrier{\ccplx^{m-1}} \cap V$ are two distinct points such
  that $F(w)=F(z)$. Since $F$ is injective on each simplex, there are
  distinct simplices $\gsplxs, \gsplxt$ such that
  $w \in \rint{\gsplxs}$ and $z \in \rint{\gsplxt}$.  So there is an
  open neighbourhood $U$ of $F(w)=F(z)$ that is contained in
  $F(\ostr{\gsplxs}) \cap F(\ostr{\gsplxt})$.

  We must have $\ostr{\gsplxs} \cap \ostr{\gsplxt} = \emptyset$,
  because if $x \in \ostr{\gsplxs} \cap \ostr{\gsplxt}$, then
  $x \in \rint{\gsplxm}$ for some $\gsplxm$ that has both $\gsplxs$
  and $\gsplxt$ as faces. But this means that both $w$ and $z$ belong
  to $\gsplxm$, contradicting the injectivity of $F|_{\gsplxm}$. It
  follows that points in the nonempty set
  $U \setminus \carrier{\ccplx^{m-1}}$ have at least two points in
  their preimage, a contradiction. Thus $F|_V$ is injective, and
  therefore, by Brouwer's invariance of domain, it is an embedding.
%
  % The same kind of reasoning that showed that $F(\str{\gsplxs})$ is
  % open, also shows that $F|_{\Omega}$ is an open map: Let
  % $W \subseteq \Omega$ be open, and $x \in W$. We need to show that
  % there is an open neighbourhood $U$ of $F(x)$ such that
  % $U \subseteq F(W)$. Let $W' = W \cap \ostr{\gsplxt}$, where
  % $x \in \rint{\gsplxt}$. We choose the neighbourhood $U$ small enough
  % that $U \cap F(\bdry{W'}) = \emptyset$, and argue as above that
  % $U \subseteq F(W')$.
%
  % It follows that $F|_V$ is an embedding.
\end{proof}

Our strategy for employing 
\Lemref{lem:splx.pos.embed}
is to demand that
the restriction of
$F_p$ 
to any $m$-simplex
has low metric distortion, and use this fact to ensure that the
image of $V_p \subset \carrier{\ccplx_p}$ is not intersected by the
image of the boundary of $\carrier{\ccplx_p}$, i.e., we will establish
that $F_p(V_p) \cap F_p(\carrier{\bdry{\ccplx_p}}) = \emptyset$. We need
to also establish that there is a point $y$ in
$F_p(V_p) \setminus F_p(\carrier{\ccplx_p^{m-1}})$ 
such that $F^{-1}(y)$ is a single
point. The metric distortion bound will help us here as well.

\begin{de}[$\xi$-distortion map]
  \label{def:distortion.map}
  A map $F\colon U\subset \rem \to \rem$ is a \defn{$\xi$-distortion
    map} if for all $x,y \in U$ we have
  \begin{equation}
    \label{eq:def.distortion.map}
    \bigl|\norm{F(x)-F(y)} - \norm{x-y}\bigr| \leq \xi\norm{x-y}.
  \end{equation}
\end{de}

We are interested in $\xi$-distortion maps with small $\xi$.
\Eqnref{eq:def.distortion.map} can be equivalently written
\begin{equation*}
  (1-\xi)\norm{x-y} \leq \norm{F(x)-F(y)} \leq (1+\xi)\norm{x-y},
\end{equation*}
and it is clear that when $\xi<1$, a $\xi$-distortion map is a
bi-Lipschitz map. For our purposes the metric distortion constant
$\xi$ is more convenient than a bi-Lipschitz constant. It is easy to
show that if $F$ is a $\xi$-distortion map, with $\xi<1$, then $F$ is
a homeomorphism onto its image, and $F^{-1}$ is a
$\frac{\xi}{1-\xi}$-distortion map (see
\Lemref{lem:inv.comp.distort}(1)).

Assuming that $F_p|_{\gsplxs}$ is a $\xi$-distortion map for each
$m$-simplex $\gsplxs \in \ccplx_p$, we can bound how much it displaces
points. Specifically, for any point $x \in \carrier{\ccplx_p}$, we
will bound $\norm{x-F(x)}$.  We exploit the fact that the $m+1$
vertices of $\gsplxs$ remain fixed, and use
\defn{trilateration},
%~\cite{hereman1995}. 
i.e., we use the estimates of the distances to the fixed vertices to
estimate the location of $F(x)$.  Here, the quality of the simplex
comes into play.

\begin{notation}[simplex quality]
\label{not:splx.qual}
If $p$ is a vertex of $\gsplxs$, the \defn{altitude} of $p$ is the
distance from $p$ to the opposing facet of $\gsplxs$ and is denoted
$a_p(\gsplxs)$. 
The \defn{thickness} of $\gsplxs$, denoted
$t(\gsplxs)$ (or just $t$ if there is no risk of confusion) is given
by $\frac{a}{mL}$, where $a=a(\gsplxs)$ is the smallest altitude of
$\gsplxs$, and $L=L(\gsplxs)$ is the length of the longest edge. We
set $t(\gsplxs)=1$ if $\gsplxs$ has dimension~$0$.
\end{notation}

\begin{lem}[trilateration]
  \label{lem:distort.trilateration}
  Suppose $\gsplxs \subset \rem$ is an $m$-simplex, and
  $F\colon\gsplxs \to \rem$ is a $\xi$-distortion map that leaves the
  vertices of $\gsplxs$ fixed. If $\xi \leq 1$, then for any
  $x \in \gsplxs$,
  \begin{align*}
    \norm{x-F(x)} \leq \frac{3\xi L}{t},
  \end{align*}
  where $L$ is the length of the longest edge of $\gsplxs$, and $t$ is
  its thickness.
\end{lem}
\begin{proof}
  Let $\{p_0,\ldots,p_m\}$ be the vertices of $\gsplxs$.
  For $x \in \gsplxs$, let $\tilde{x} = F(x)$. 

  We choose $p_0$ as the origin, and observe that
  \begin{equation}
    \label{eq:position.component}
    \transp{p_i}x = \frac{1}{2}\left( \norm{x}^2 + \norm{p_i}^2 -
      \norm{x-p_i}^2 \right), 
  \end{equation}
  which we write in matrix form as $\transp{P}x = b$, where $P$ is the
  $m \times m$ matrix whose $i$-th column is $p_i$, and $b$ is the
  vector whose $i$-th component is given by the right-hand side of
  \eqref{eq:position.component}. Similarly, we have
  $\transp{P}\tilde{x} = \tilde{b}$ with the obvious definition of
  $\tilde{b}$. Then
\[
\tilde{x} - x = (\transp{P})^{-1}(\tilde{b} - b).
\]

  Since $F(p_0) = p_0 = 0$, we have
  $\abs{\norm{\tilde{x}} - \norm{x}} \leq \xi \norm{x}$, and so
\[
\abs{\norm{\tilde{x}}^2 - \norm{x}^2} 
\leq \xi (2 + \xi) \norm{x}^2 \leq 3 \xi L^2.
\]
  Similarly,
  $\abs{ \norm{x-p_i}^2 - \norm{\tilde{x} - p_i}^2 } < 3 \xi
  L^2$.
  Thus $\abs{\smash{\tilde{b}_i} - b_i} \leq 3\xi L^2$, and
  $\norm{\smash{\tilde{b}} - b} \leq 3\sqrt{m}\xi L^2$.

  By \cite[Lemma~2.4]{boissonnat2013stab1} %Lemref{lem:bound.skP}
  we have $\onorm{ (\transp{P})^{-1} } \leq (\sqrt{m} t L)^{-1}$, and
  the stated bound follows.
\end{proof}

\subsubsection{Using $\protect\str{p}$ as $\tccplx_p$}
\label{sec:star.method}

For the local complex $\tccplx_p \subset \acplx$ introduced in
\Defref{def:compatible.atlases}, we now make a specific choice:
$\tccplx_p =\str{p}$. This is the smallest complex allowed by the
definition.  For convenience, we define $\p=\lmap_p(p)$, so that
$\ccplx_p = \lmap_p(\tccplx_p) = \str{\p}$.

\begin{figure}
  \centering
  \includegraphics[scale=.3]{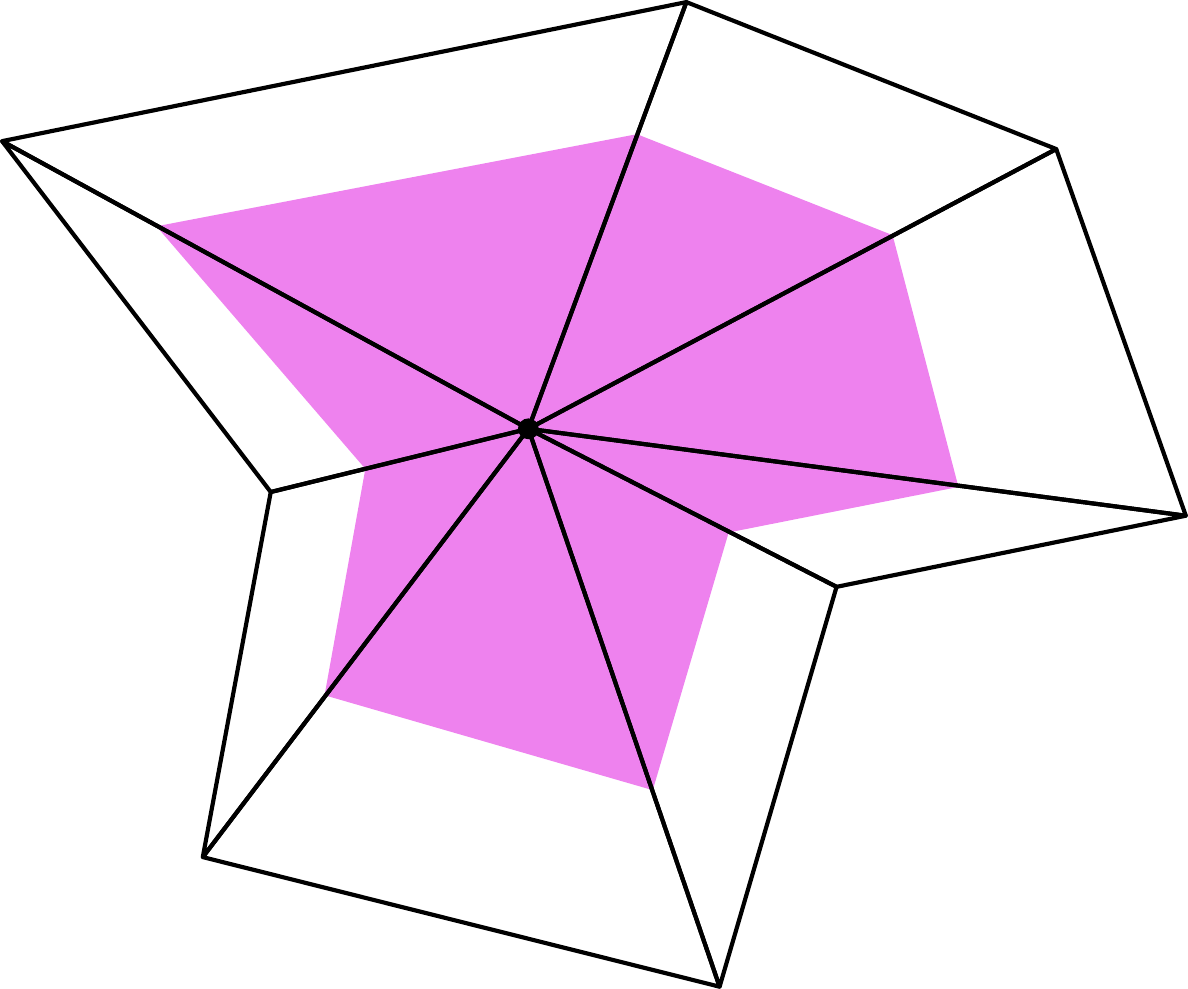}
  \caption{The open set $V_p$ (shaded) is a homothetically shrunk copy
  of the interior of $\strp$. }
  \label{fig:shrunk.star}
\end{figure}

We define $V_p$ to be the open set obtained by homothetically
``shrinking'' $\carrier{\strp}$ such that it is just large enough to
contain the barycentres of the simplices that have $\p$ as a vertex
(see Figure~\ref{fig:shrunk.star}).
To be more specific we define $V_p$ to be the open set consisting of
the points in $\carrier{\strp}$ whose barycentric coordinate with
respect to $\p$ is strictly larger than $\frac{1}{m+1}-\delta$, where
$\delta>0$ is arbitrarily small.  Since the barycentric coordinates in
each $m$-simplex sum to 1, and the piecewise linear maps $\lmap_p$
preserve barycentric coordinates, this ensures that the sets
$\lmap_p^{-1}(V_p)$ cover $\carrier{\acplx}$.

% We assume that $F_p$ is a $\xi$-distortion map on each simplex.  The
% idea is to show that $F_p$ is an embedding on a set $V_p$ that
% includes all the points $x \in \carrier{\strp}$ such that the
% barycentric coordinate of $x$ associated with $\p$ in an $m$-simplex
% that contains $x$ is at least $\frac{1}{m+1}$. (That is the homothetic
% copy of $\carrier{\strp}$, ``shrunk'' by a factor of
% $1-\frac{1}{m+1}$, using $\p$ as the origin.)  To be more specific we
% define $V_p$ to be the open set consisting of the points in
% $\carrier{\strp}$ whose barycentric coordinate with respect to $\p$ is
% strictly larger than $\frac{1}{m+1}-\delta$, where $\delta>0$ is
% arbitrarily small.  Since the barycentric coordinates in each
% $m$-simplex sum to 1, and the piecewise linear maps $\lmap_p$ preserve
% barycentric coordinates, this ensures that the sets
% $\lmap_p^{-1}(V_p)$ cover $\carrier{\acplx}$.

We assume that $F_p$ is a $\xi$-distortion map on each simplex.  The
idea is to show that $F_p$ is an embedding on $V_p$.  In order to
employ the simplexwise positive embedding lemma
(\Lemref{lem:splx.pos.embed}), we need to establish that there is a
point in $V_p \setminus \carrier{\ccplx_p^{m-1}}$ that is not mapped
to the image of any other point in $\carrier{\ccplx_p}$. We choose the
barycentre of a simplex for this purpose. We say that a simplicial
complex is a \defn{pure} $m$-dimensional simplicial complex if every
simplex is the face of an $m$-simplex.

\begin{lem}[a point covered once]
  \label{lem:pt.covered.once}
  Suppose $\ccplx$ is a pure $m$-dimensional finite simplicial complex embedded
  in $\rem$, and that for each $\gsplxs \in \ccplx$ we have
  $t(\gsplxs) \geq t_0$.  
  Let $\gsplxs \in \ccplx$ be an $m$-simplex with the largest diameter,
  i.e., $L(\gsplxs) \geq L(\gsplxt)$ for all $\gsplxt \in
  \ccplx$, and let $b$ be the barycentre of $\gsplxs$.  
If $F \colon \carrier{\ccplx} \to \rem$ 
  leaves the vertices of $\ccplx$ fixed, and its restriction to any
  $m$-simplex in $\ccplx$ is a $\xi$-distortion map with
  \begin{equation}
    \label{eq:bary.bnd.required}
    \xi \leq \frac16\frac{m}{m+1}t_0^2,
  \end{equation}
  then $F^{-1}(F(b)) = \{b\}$.
\end{lem}

\begin{proof}
  Since $\xi < 1$, the restriction of $F$ to $\gsplxs$ is
  injective. Suppose $x\in \carrier{\ccplx}$ is such that $F(x)=F(b)$, but
  $x\neq b$. Then $x$
  belongs to some $m$-simplex $\gsplxt \in \ccplx$ different from
  $\gsplxs$. Since the distance from $b$ to the boundary of $\gsplxs$
  is
  \[
    \frac{a(\gsplxs)}{m+1} = \frac{mt(\gsplxs)L(\gsplxs)}{m+1},
  \]
  it follows that $\norm{x-b}> \unpfrac{mt(\gsplxs)L(\gsplxs)}{m+1}$.
But
  using \Lemref{lem:distort.trilateration} and the constraint
\eqref{eq:bary.bnd.required} we arrive at a contradiction:
  \begin{equation*}
    \norm{x-b}\leq \norm{b-F(b)} + \norm{x-F(x)}
    \leq
    \frac{3\xi L(\gsplxs)}{t(\gsplxs)} 
    +  \frac{3\xi L(\gsplxt)}{t(\gsplxt)} 
    \leq  \frac{mt(\gsplxs)L(\gsplxs)}{m+1}.
  \end{equation*}
\end{proof}

Now we also need to ensure that
$F_p(V_p) \cap F_p(\carrier{\bdry{\ccplx_p}}) = \emptyset$. Here we
will explicitly use the assumption that $\ccplx_p$ is $\strp$.  We
say that $\strp$ is a \defn{full star} if its carrier is an
$m$-manifold with boundary and $\p$ does not belong to
$\bdry{\strp}$.

\begin{lem}[barycentric boundary separation]
  \label{lem:dist.to.star.bdry}
  Suppose $\strp$ is a full $m$-dimen\-sional star embedded in
  $\rem$. Let $a_0 = \min_{\gsplxs \in \strp}a_{\p}(\gsplxs)$ be the
  smallest altitude of $\p$ in the $m$-simplices in $\strp$. Suppose
  $x \in \gsplxs \in \strp$, where $\gsplxs$ is an $m$-simplex, and
  $\lambda_{\gsplxs,\p}(x)$, the barycentric coordinate of $x$ with
  respect to $\p$ in $\gsplxs$, satisfies
  $\lambda_{\gsplxs,\p}(x) \geq \alpha$. Then
  $\distEm{x}{\carrier{\bdry{\strp}}} \geq \alpha a_0$.

  If $t_0$ is a lower bound on the thicknesses of the simplices in
  $\strp$, and $s_0$ is a lower bound on their diameters, then
  $\distEm{x}{\carrier{\bdry{\strp}}} \geq \alpha m t_0 s_0$.
\end{lem}

\begin{proof}
  Since we are interested in the distance to the boundary, consider a
  point $y \in \carrier{\bdry{\strp}}$ such that the segment $[x,y]$
  lies in $\carrier{\strp}$. The segment passes through a sequence of
  $m$-simplices, $\gsplxs_0=\gsplxs,\gsplxs_1, \dots, \gsplxs_n$, that
  partition it into subsegments $[x_i,y_i] \subset \gsplxs_i$ with
  $x_0=x$,\, $y_n=y$ and $x_i=y_{i-1}$ for all $i\in \{1,\dots,n\}$.

  Observe that
  $\lambda_{\gsplxs_i,\p}(x_i)=\lambda_{\gsplxs_{i-1},\p}(y_{i-1})$, and
  that
\[
\norm{x_i-y_i}\geq
  a_{\p}(\gsplxs_i)\abs{\lambda_{\gsplxs_i,\p}(x_i)
    -\lambda_{\gsplxs_i,\p}(y_i)}.
\]
  Thus
  \begin{equation*}
    \begin{split}
      \norm{x-y}&= \sum_{i=0}^n \norm{x_i-y_i}\\
      &\geq \sum_{i=0}^n a_{\p}(\gsplxs_i)
      \abs{\lambda_{\gsplxs_i,\p}(x_i)-\lambda_{\gsplxs_i,\p}(y_i)}\\
      &\geq a_0\sum_{i=0}^n 
      (\lambda_{\gsplxs_i,\p}(x_i)-\lambda_{\gsplxs_i,\p}(y_i))\\
      &= a_0 (\lambda_{\gsplxs,\p}(x)-\lambda_{\gsplxs_n,\p}(y))
      = a_0\lambda_{\gsplxs,\p}(x)\\
      &\geq a_0 \alpha.
    \end{split}
  \end{equation*}

  From the definition of thickness we find that $a_0\geq t_0 m s_0$,
  yielding the second statement of the lemma.
\end{proof}

\Lemref{lem:dist.to.star.bdry} allows us to quantify the
distortion bound that we need to ensure that the boundary of $\strp$
does not get mapped by $F_p$ into the image of the open set $V_p$. The
argument is the same as for \Lemref{lem:pt.covered.once}, but there we
were only concerned with the barycentre of the largest simplex, so the
relative sizes of the simplices were not relevant as they are here
(compare the bounds \eqref{eq:bary.bnd.required} and
\eqref{eq:bdry.bnd.required}).

\begin{lem}[boundary separation for $V_p$]
  Suppose $\strp$ is a full star embedded in $\R^m$, and every
  $m$-simplex $\gsplxs$ in $\strp$ satisfies
  $s_0 \leq L(\gsplxs) \leq L_0$, and $t(\gsplxs)\geq t_0$. If the
  restriction of $F_p$ to any $m$-simplex in $\strp$
  is a $\xi$-distortion map, with
  \begin{equation}
    \label{eq:bdry.bnd.required}
    \xi < \frac{1}{6}\frac{m}{m+1}\frac{s_0}{L_0}t_0^2,
  \end{equation}
  then $F_p(V_p) \cap F_p(\carrier{\bdry{\strp}}) = \emptyset$,
  where $V_p$ is the set of points with barycentric coordinate with
  respect to $\p$ in a containing $m$-simplex strictly greater than
  $\frac{1}{m+1} - \delta$, with 
% $\delta$  an arbitrarily small
%   parameter satisfying
%   \begin{equation}
%     \label{eq:arb.small}
%     0 < \delta \leq \frac{1}{m+1} - \frac{6L_0\xi}{m s_0 t_0^2}.
%   \end{equation}
$\delta>0$ an arbitrary, suffiently small parameter.
\end{lem}

\begin{proof}
  If $x \in \gsplxs \in \strp$ has barycentric coordinate with
  respect to $\p$ larger than $\frac{1}{m+1}-\delta$, and
  $y \in \gsplxt \in \bdry{\strp}$, then Lemmas
  \ref{lem:distort.trilateration} and \ref{lem:dist.to.star.bdry}
  ensure that $F_p(x) \neq F_p(y)$ provided
  \begin{equation*}
    \frac{3\xi L(\gsplxs)}{t(\gsplxs)} 
    +  \frac{3\xi L(\gsplxt)}{t(\gsplxt)} 
    \leq  \left(\frac{1}{m+1}-\delta\right)m s_0t_0,
  \end{equation*}
  which is satisfied by \eqref{eq:bdry.bnd.required} when 
%$\delta$ satisfies \eqref{eq:arb.small}.
  $\delta>0$ satisfies
  \begin{equation*}
    \delta \leq \frac{1}{m+1} - \frac{6L_0\xi}{m s_0 t_0^2}.
  \end{equation*}
\end{proof}

When inequality \eqref{eq:bdry.bnd.required} (and therefore
also inequality \eqref{eq:bary.bnd.required}) is satisfied, we can
employ the embedding lemma (\Lemref{lem:splx.pos.embed}) to guarantee
that $V_p$ is embedded:

\begin{lem}[local homeomorphism]
  \label{lem:local.homeo}
  Suppose $\acplx$ is a compact $m$-manifold complex (without
  boundary), with vertex set $\pts$, and $\man$ is an $m$-manifold. A
  map $H \colon \carrier{\acplx} \to \man$ is a 
%  local homeomorphism
  covering map if the following criteria are satisfied:
  \begin{enumerate}
  \item \textbf{\textup{compatible atlases}}\, There are compatible
    atlases for $H$, with
    $\tccplx_p = \str{p}$ for each $p \in \pts$
    {\rm(\Defref{def:compatible.atlases})}.
  \item \textup{\bf simplex quality}\, For each $p \in \pts$, every
    simplex $\gsplxs \in \strp = \lmap_p(\str{p})$ satisfies
    $s_0 \leq L(\gsplxs) \leq L_0$ and $t(\gsplxs) \geq t_0$
    {\rm (Notation~\ref{not:splx.qual})}.
  \item \textup{\bf distortion control}\, For each $p\in \pts$, the map
    \[
      F_p = \phi_p \circ H \circ \lmap_p^{-1}
      \colon
      \carrier{\strp} \to \rem,
    \]
    when restricted to any $m$-simplex in $\strp$,
    is an orientation-preserving $\xi$-distortion map with
    \[
      \xi < \frac{m s_0 t_0^2}{6(m+1)L_0}
    \]
    {\rm(Definitions \ref{def:orientation.preserving} and
    \ref{def:distortion.map})}.
  \end{enumerate}
\noproof
\end{lem}

\subsection{Injectivity}
\label{sec:injectivity}

Having established that $H$ is a covering map, to ensure that $H$ is
injective it suffices to demonstrate that on each component of $\man$
there is a point with only a single point in its preimage. Injectivity
follows since the number of points in the preimage is locally constant
for covering maps.

Since each simplex is embedded by $H$, it is sufficient to show that
for each vertex $q \in \pts$, if $H(q) \in H(\gsplxs)$, then $q$ is a
vertex of $\gsplxs$. This ensures that $H^{-1}(H(q)) = \{q\}$, and
%\Defref{def:compatible.atlases} ensures that 
by \Remref{rem:touch.all.components} each component of $\man$ must
contain the image of a vertex.

In practice, we typically don't obtain this condition directly. The
complex $\acplx$ is \emph{constructed} by means of the local patches
$\ccplx_p$, and it is with respect to these patches that the vertices
behave well. 

\begin{de}[vertex sanity]
  \label{def:vertex.sanity}
  If $H \colon \carrier{\acplx} \to \man$ has compatible atlases
  (\Defref{def:compatible.atlases}), then $H$ exibits \defn{vertex
    sanity} if: for all vertices $p,q \in \pts$, if
  $\phi_p \circ H(q) \in \carrier{\strp} =
  \lmap_p(\carrier{\str{p}})$, then $q$ is a vertex of~$\str{p}$.
\end{de}

Together with the distortion bounds that are imposed on $F_p$,
\Defref{def:vertex.sanity} ensures that the image of a vertex cannot
lie in the image of a simplex to which it does not belong:

\begin{lem}[injectivity]
  \label{lem:injectivity}
  If $H \colon \carrier{\acplx} \to \man$ satisfies the hypotheses of
  \Lemref{lem:local.homeo} as well as \Defref{def:vertex.sanity}, then
  $H$ is injective, and therefore a homeomorphism.
\end{lem}

\begin{proof}
  % Observe that for any simplex $\gsplxs \in \acplx$, the piecewise
  % linear coordinate maps $\lmap_p$ preserve the barycentric
  % coordinates (by definition).
  Towards a contradiction, suppose that $H(q) \in H(\gsplxs)$ and that
  $q$ is not a vertex of the $m$-simplex $\gsplxs$. This means there
  is some $x \in \gsplxs$ such that $H(x) = H(q)$. Let $p$ be a vertex
  of~$\gsplxs$.  The vertex sanity hypothesis
  (\Defref{def:vertex.sanity}) implies that $\phi_p \circ H(q)$ must
  be either outside of $\carrier{\strp}$, or belong to its boundary.
  Thus Lemmas \ref{lem:dist.to.star.bdry} and
  \ref{lem:distort.trilateration}, and the bound on $\xi$ from
  \Lemref{lem:local.homeo}(3) imply that the barycentric coordinate of
  $x$ with respect to $p$ must be smaller than $\frac{1}{m+1}$: Let
  $\x = \lmap_p(x)$, and $\hat{\gsplxs}=\lmap_p(\gsplxs)$.
  \Lemref{lem:distort.trilateration} says that
  \begin{equation*}
    \abs{F_p(\x)-\x} \leq \frac{3\xi L_0}{t_0} < \frac{m s_0 t_0}{2(m+1)}
    \leq \frac{a_0}{2(m+1)},
  \end{equation*}
  where $a_0$ is a lower bound on the altitudes of $\p$, as in
  \Lemref{lem:dist.to.star.bdry}. Since $F_p(\x)=\phi_p \circ H(x)$ is
  at least as far away from $\x$ as $\bdry{\strp}$,
  \Lemref{lem:dist.to.star.bdry} implies that the barycentric
  coordinate of $\x \in \hat{\gsplxs}$ with respect to $\p$ must be
  no larger than $\frac{1}{2(m+1)}$. Since $\lmap_p$ preserves
  barycentric coordinates, and the argument works for any vertex $p$
  of $\gsplxs$, we conclude that all the barycentric coordinates of
  $x$ in $\gsplxs$ are strictly less than $\frac{1}{m+1}$. We have reached a
  contradiction with the fact that the barycentric coordinates
  of $x$ must sum to 1.
\end{proof}

\subsection{Main result}

To recap, Lemmas \ref{lem:local.homeo} and \ref{lem:injectivity} yield
the following triangulation result.  In the bound on $\xi$ from
\Lemref{lem:local.homeo}(3), we replace the factor $\frac{m}{m+1}$
with $\frac12$, the lower bound attained when $m=1$.

\begin{thm}[triangulation]
  \label{thm:metric.triang}
  Suppose $\acplx$ is a compact $m$-manifold complex (without
  boundary), with vertex set $\pts$, and $\man$ is an $m$-manifold. A
  map $H \colon \carrier{\acplx} \to \man$ is a homeomorphism if the
  following criteria are satisfied:
  \begin{enumerate}
  \item \textbf{\textup{compatible atlases}}\, There are compatible
    atlases
\[
\{(\tccplx_p,\lmap_p)\}_{p\in\pts}, \quad \tccplx_p \subset \acplx, 
\quad and \quad
\{(U_p,\phi_p)\}_{p\in\pts}, \quad U_p \subset \man,
\]
    for $H$, with
    $\tccplx_p = \str{p}$ for each $p \in \pts$, the vertex set of~$\acplx$
    {\rm(\Defref{def:compatible.atlases})}.
  \item \textup{\bf simplex quality}\, For each $p \in \pts$, every
    simplex $\gsplxs \in \strp = \lmap_p(\str{p})$ satisfies
    $s_0 \leq L(\gsplxs) \leq L_0$ and $t(\gsplxs) \geq t_0$ {\rm
      (Notation~\ref{not:splx.qual})}.
  \item \textup{\bf distortion control}\, For each $p\in \pts$, the map
    \[
      F_p = \phi_p \circ H \circ \lmap_p^{-1}
      \colon
      \carrier{\strp} \to \rem,
    \]
    when restricted to any $m$-simplex in $\strp$,
    is an orientation-preserving $\xi$-distortion map with
    \[
      \xi < \frac{s_0 t_0^2}{12 L_0}
    \]
    {\rm(Definitions \ref{def:orientation.preserving} and
    \ref{def:distortion.map})}.
  \item \textup{\bf vertex sanity}\, For all vertices $p,q \in \pts$,
    if $\phi_p \circ H(q) \in \carrier{\strp}$, then $q$ is a vertex
    of~$\str{p}$.
  \end{enumerate}
\noproof
\end{thm}

\begin{remark}
  The constants $L_0$, $s_0$, and $t_0$ that constrain the simplices
  in the local complex $\strp$, and the metric distortion of $F_p$ in
  \Thmref{thm:metric.triang} can be considered to be local, i.e., they
  may depend on $p\in\pts$.
\end{remark}

%%% Local Variables:
%%% mode: latex
%%% TeX-master: "homeo"
%%% End:

%% file: distortion.tex
\section{Metric and differentiable distortion maps}
\label{sec:distortion}

In this short section we gather some useful lemmas on general
distortion maps, and on differentiable distortion maps.

\subsection{Metric distortion maps}
\label{sec:metric.distortion}

The definition of a distortion map (\Defref{def:distortion.map}) makes 
sense in a more general context: a
map $F: (X,d_X) \to (Y,d_Y)$ between metric spaces is a
\defn{$\xi$-distortion map} if
\begin{equation}
\label{eq:generic.distortion.map}
  \abs{d_Y(F(x),F(y)) - d_X(x,y)} \leq \xi d_X(x,y) \quad \text{for
    all }x,y \in X.
\end{equation}

\begin{lem}[inverse and composition of distortion maps]
  \label{lem:inv.comp.distort}
  \begin{enumerate}
    \item  If $F: (X,d_X) \to (Y,d_Y)$ is a
      $\xi$-distortion map with $\xi < 1$, then $F^{-1}$ is a
      $\frac{\xi}{1-\xi}$-distortion map.
    \item  Suppose $F_i:
      (X_i,d_{X_i})\to(X_{i+1}, d_{X_{i+1}})$,\, $1\leq i \leq k$, are
      respectively $\xi_i$-distortion maps. Then
      \begin{equation*}
        F_k \circ F_{k-1} \circ \cdots \circ F_1: X_1 \to X_{k+1}
      \end{equation*}
      is a $\bigl(\sum_{W \in \npwr{\{k\}}} \prod_{i\in W} \xi_i
      \bigr)$-distortion map, where $\npwr{\{k\}}$ is the set of nonempty
      subsets of $\{1, \ldots, k\}$.

      In particular, $F_2\circ F_1$ is a
      $(\xi_1{+}\xi_2{+}\xi_1\xi_2)$-distortion map.
  \end{enumerate}
\end{lem}

\begin{proof}
  (1)\, Let $u=F(x)$ and $v=F(y)$. Then
  \eqref{eq:generic.distortion.map} becomes
  \begin{equation}
    \label{eq:raw.inv.distort}
    \abs{d_X(F^{-1}(u),F^{-1}(v)) - d_Y(u,v)} \leq \xi
    d_X(F^{-1}(u),F^{-1}(v)). 
  \end{equation}
  But it follows from \eqref{eq:raw.inv.distort} that
  $d_X(F^{-1}(u),F^{-1}(v)) \leq \frac{1}{1-\xi}d_Y(u,v)$, and
  plugging this back into \eqref{eq:raw.inv.distort} yields the result.

  \smallbreak\noindent
(2)\, Using the observation that
\begin{equation*}
  d_{X_2}(F_1(x),F_1(y)) \leq (1 + \xi_1)d_{X_1}(x,y),
\end{equation*}
we find
\begin{equation*}
  \begin{split}
    &\abs[big]{d_{X_3}(F_2\circ F_1(x), F_2\circ F_1(x)) - d_{X_1}(x,y)}
\\
% &\phantom{d_{X_3}F_2\circ}
% = \bigl\lvert d_{X_3}(F_2\circ F_1(x), F_2\circ F_1(x))
%       - d_{X_2}(F_1(x), F_1(y))
% \\
%   &\hskip220pt
% + d_{X_2}(F_1(x), F_1(y))- d_{X_1}(x,y) \bigr\rvert \\
&\phantom{d_{X_3}F_2\circ}
\leq \abs[big]{ d_{X_3}(F_2\circ F_1(x), F_2\circ F_1(x))
      - d_{X_2}(F_1(x), F_1(y))}
\\
&\hskip220pt
+ \abs[big]{d_{X_2}(F_1(x), F_1(y))- d_{X_1}(x,y)} \\
&\phantom{d_{X_3}F_2\circ}
\leq \xi_2 d_{X_2}(F_1(x), F_1(y))
+ \xi_1 d_{X_1}(x,y) \\
&\phantom{d_{X_3}F_2\circ}
\leq \xi_2 (1 + \xi_1)d_{X_1}(x,y)
+ \xi_1 d_{X_1}(x,y) \\
&\phantom{d_{X_3}F_2\circ}
= (\xi_1 +\xi_2  + \xi_2\xi_1)d_{X_1}(x,y).
  \end{split}
\end{equation*}
This establishes the bound for the composition of two maps, but it
also serves as the inductive step for the general bound. If $G$ is the
$\eta$-distortion map defined by $G = F_{k-1}\circ \cdots \circ F_1$,
then $F_k \circ G$ is a $(\xi_k {+} \eta {+} \xi_k \eta)$-distortion
map. Now notice that 
\begin{equation*}
\sum_{W \in \npwr{\{k\}}} \prod_{i\in W} \xi_i
=
\xi_k +   \left(\sum_{W \in \npwr{\{k-1\}}} \prod_{i\in W} \xi_i \right)
+ \xi_k \left(\sum_{W \in \npwr{\{k-1\}}} \prod_{i\in W} \xi_i \right).
\end{equation*}
\end{proof}

\subsection{Differentiable distortion maps}
\label{sec:differentiable.distortion}

Although the homeomorphism demonstration that yields
\Thmref{thm:metric.triang} makes no explicit requirement of
differentiability, it is convenient to exploit differentiability when
it is available. We collect here some observations relating metric
distortion and bounds on the differential of a map.

Recall that if $F\colon U\subseteq \R^m \to \R^m$ is differentiable,
then the \defn{differential} of $F$ at $x$ is the linear map defined
by
\begin{equation*}
  dF_x(v) = \frac{d}{dt}F\circ \alpha(t)\big|_{t=0},
\end{equation*}
where $\alpha \colon I\subset \R \to \R^m$ is any curve such that
$\alpha(0)=x$ and $\alpha'(0) = v$, where $\alpha'$ is the derivative
with respect to $t$.

Bounds on the differential of $F$ are closely related to the metric
distortion of $F$. If $A\colon U\subseteq \R^m \to \R^m$ is a linear
map, we let $\onorm{A}$ denote the operator norm:
$\onorm{A} = \sup_{\norm{v}=1}\norm{Av}$. Associated with $A$ are $m$
nonnegative numbers called the \defn{singular values} of $A$, denoted
$s_i(A)$,\, $1\leq i \leq m$, ordered such that $s_i(A) \geq s_j(A)$
if $i\leq j$.  We only mention the singular values because they
provide a notational convenience. The largest singular value is
defined by $s_1(A)=\onorm{A}$, and the smallest singular value is
$s_m(A)=\inf_{\norm{v}=1}\norm{Av}$. 

\begin{lem}
  \label{lem:distort.bnd.differential}
  If $F\colon U \subseteq \R^m \to \R^m$ is a differentiable
  $\xi$-distortion map, then
  \begin{equation*}
    \abs{s_i(dF_x)-1} \leq \xi
    \quad \text{for all }x\in U \text{ and } 1 \leq i \leq m.
  \end{equation*}
\end{lem}

\begin{proof}
  For $x\in U$, and $v \in T_x\R^m = \R^m$ with $\norm{v}=1$, let 
$\alpha(t)= x + tv$. Since $F$ is a $\xi$-distortion map, we have
\begin{equation*}
  (1-\xi)\norm{(x+tv)-x} \leq \norm{F(x+tv)-F(x)} \leq (1+\xi)\norm{(x+tv)-x},
\end{equation*}
so for all $t\neq 0$,
\begin{equation*}
  1-\xi \leq \frac{\norm{F(x+tv) - F(x)}}{\abs{t}} \leq 1+\xi.
\end{equation*}
Since $F$ is differentiable, 
\begin{equation*}
\lim_{t\to0}\frac{\norm{F(x+tv) - F(x)}}{\abs{t}}
=  \norm{\lim_{t\to0}\frac{F(x+tv) - F(x)}{t}}
= \norm{dF_x(v)},
\end{equation*}
so
\begin{equation*}
  1-\xi \leq \norm{dF_x(v)} \leq 1+\xi,
\end{equation*}
which yields the claimed result.
\end{proof}

So a bound on the metric distortion of a differentiable map directly
yields the same bound on the amount the singular values of the
differential can differ from $1$. We are interested in a converse
assertion: we want to bound the metric distortion of $F$, given a
bound on (the singular values of) its differential. This can only be
done with caveats.

\begin{lem}
  \label{lem:differential.bnd.distort}
  Suppose $\conset$ is a convex set, $\conset \subseteq U \subseteq \rem$, and
  $F\colon U \to \rem$ is a differentiable map such that
  $F(\conset) \subseteq V \subseteq F(U)$ for some convex set $V$.
  If $F$ is injective, and
  \begin{equation*}
    \abs{s_i(dF_x)-1} \leq \xi < 1, \quad \text{for all }
    x\in U \text{ and } 1 \leq i \leq m,
  \end{equation*}
  then $F|_{\conset}$ is a $\xi$-distortion map.
\end{lem}

\begin{proof}
  Since $\conset$ is convex, the length of the image of the line segment
  between $x$ and $y$ provides an upper bound on the distance between
  $F(x)$ and $F(y)$, and this length can be bounded because  of the
  bound on $dF$:
  \begin{equation}
    \label{eq:up.bnd.length}
    \norm{F(y) - F(x)} \leq \int_0^1 \norm{dF_{x +t(y-x)}(y-x)} \, dt
    \leq (1 +\xi)\norm{y-x}. 
  \end{equation}
%  Thus $\norm{F(y) - F(x)} - \norm{y-x} \leq \xi\norm{y-x}$.

  To get a lower bound we use the fact that injectivity and the bound
  on the singular values imply that $F$ is invertible. For any point
  $z=F(x)$ we have
  \begin{equation*}
    \onorm{dF^{-1}_z} = s_m(dF_x)^{-1} \leq (1-\xi)^{-1}.
  \end{equation*}
  Since, for $x,y\in\conset$ the segment $[F(x),F(y)]$ is contained in
  $F(U)$, we can use the same argument as in \eqref{eq:up.bnd.length},
  but using $F^{-1}$ instead of $F$, so
  \begin{equation*}
    \norm{x-y} \leq (1-\xi)^{-1} \norm{F(x)-F(y)}.
  \end{equation*}
  Therefore, combining with \eqref{eq:up.bnd.length} we have
  \begin{equation*}
    (1-\xi)\norm{x-y} \leq \norm{F(x)-F(y)}
      \leq (1+\xi) \norm{x-y}.
  \end{equation*}
\end{proof}

%%% Local Variables:
%%% mode: latex
%%% TeX-master: "homeo"
%%% End:

%% file: subman.tex
\section{Submanifolds of Euclidean space}
\label{sec:subman}

As a specific application of \Thmref{thm:metric.triang}, we consider a
smooth (or at least $C^2$) compact $m$-dimensional submanifold of
Euclidean space: $\man \subset \R^N$. A simplicial complex $\acplx$ is
built whose vertices are a finite set $\pts$ sampled from the manifold:
$\pts \subset \man$. The motivating model for this setting is the
tangential Delaunay complex \cite{boissonnat2014tancplx}. In that case
$\acplx$ is constructed as a subcomplex of a weighted Delaunay
triangulation of $\pts$ in the ambient space $\R^N$, so it is necessarily
embedded. However, in general we do not need to assume \emph{a priori}
that $\acplx$ is embedded in $\R^N$. 
(This does not force us to consider $\acplx$ to be abstract in the
combinatorial sense. In particular, the simplices are Euclidean
simplices, not just sets of vertices.)
Instead, we assume only that the
embedding of the vertex set $\pts \hookrightarrow \R^N$ defines an
\defn{immersion} $\iota\colon\carrier{\acplx} \to \R^N$. By this we
mean that for any vertex $p\in \pts$ we have that the restriction of
$\iota$ to $\carrier{\str{p}}$ is an embedding.

At each point $x\in \man$, the tangent space $T_x\man \subset T_x\amb$
is naturally viewed as an $m$-dimensional affine flat in $\R^N$, with
the vector-space structure defined by taking the distinguished point
$x$ as the origin.  The maps involved in \Thmref{thm:metric.triang}
will be defined by projection maps. The coordinate charts are defined
using the orthogonal projection $ \projp\colon \R^N \to T_p\man.  $ As
discussed in \Secref{sec:orthog.proj}, for a sufficiently small
neighbourhood $U_p \subset \man$, we obtain an embedding
\begin{equation*}
  \phi_p = \projp|_{U_p}
  \colon
  U_p \subset \man \to T_p\man \cong \R^m,
\end{equation*}
which will define our coordinate maps for $\man$. 

For the map $H\colon \carrier{\acplx} \to \man$, we will employ the
closest point projection map defined in \Secref{sec:subman.geom} and
discussed further in 
\Secref{sec:dist.closest.pt.proj}.
There is an open
neighbourhood $\UM \subset \R^N$ of $\man$ on which each point
has a unique closest point on $\man$, so the closest-point projection
map $\projM\colon \UM \to \man$ is well-defined. We define
$H = \projM \circ \iota$.

As demanded by \Defref{def:compatible.atlases}, for each $p\in \pts$
the coordinate map $\lmap_p$ for $\acplx$ is the secant map of
$\phi_p \circ H$ restricted to $\tccplx_p = \carrier{\str{p}}$, and
since $\projp$ is already a linear map, and $\projM$ is the identity on the
vertices, this means
$\lmap_p = \projp \circ \iota|_{\carrier{\str{p}}}$.

In Sections \ref{sec:subman.geom} and \ref{sec:flats.n.angles} we
review some of the geometric concepts and standard results that we
will use in the rest of the section. 
% We take advantage of the fact
% that $\man$ has a natural Riemmannian structure defined by its
% embedding in $\R^N$, and exploit concepts such as the shape opertor
% and the second fundamental form. 
% However, 
In order to bound the metric
distortion of the maps $F_p = \phi_p \circ H \circ \lmap_p^{-1}$ via
\Lemref{lem:inv.comp.distort}, we are free to choose any convenient
metric on $\man$. As is common in computational geometry, we employ
here the metric of the ambient space $\R^N$, rather than the intrinsic
metric of geodesic distances.

\subsection{Submanifold geometry}
\label{sec:subman.geom}

Since $\man \subset \amb$ is compact, for any $x \in \amb$ there is a
point $z \in \man$ that realizes the distance to $\man$, i.e.,
\begin{equation*}
  \dtoM{x} :=
\distEN{x}{\man}
:= \inf_{y\in \man}\distEN{x}{y} = \distEN{x}{z}.
\end{equation*}

The \defn{medial axis} of $\man$ is the set of points
$\ax(M) \subset\amb$ that have more than one such closest point on
$\man$. In other words, if $x \in \ax(M)$, then an open ball
$\ballEN{x}{r}$, centred at $x$ and of radius $r=\dtoM{x}$, will be
tangent to $\man$ at two or more distinct points. The \defn{cut locus}
of $\man$ is the closure of the medial axis, and is denoted
$\claxM$. 
The \defn{reach} of $\man$ is defined by $\rchM :=
\distEN{\man}{\claxM}$. 
We observe below that
for compact $C^2$ submanifolds, $\rchM>0$.
Thus, by definition, every point $x$ in the open
neighbourhood $\UM := \amb \setminus \claxM$ of $\man$ has a unique
closest point $\px \in \man$. The \defn{closest-point projection} map
\begin{equation*}
  \projM\colon \UM \to \man,
\end{equation*}
takes $x$ to this closest point: $\projM(x)=\px$.

To each point $x\in \man$, we associate a \defn{normal space}
\[
N_x\man
= \{n\in T_x\amb \mid \dotprod{n}{v}= 0 \ \forall v\in T_x\man\}
\]
of vectors orthogonal to $T_x\man$. Thus
$T_x\amb = N_x\man \oplus T_x\man$. As with the tangent space, the
normal space at $x$ is naturally identified with an affine flat in
$\amb$. It has dimension $k=N-m$ and is orthogonal to $T_x\man$.
% The distinguished point $x$ serves as the
% origin for the vector-space structure.

%\subsubsection{The tubular neighbourhood theorem} 

The tubular neighbourhood theorem is a well known result in
differential topology. The statement presented here is adapted from
\cite[Theorem~4.8(13)]{federer1959}, and the regularity assertions in
the second paragraph are demonstrated in \cite{foote1984}.

\begin{thm}[tubular neighbourhood]
  \label{thm:tube.nbrhd.thm}
  There is a natural structure on $\UM$ given by partitioning it into
  subsets of points that all project via $\projM$ onto the same point
  of $\man$. This allows us to identify $\UM$ as a portion of the
  \defn{normal bundle} of $\man$,
\begin{equation*}
  \NM := \{(x,n) \in \R^N \times \R^N \mid x \in \man,\ n \in
  N_x\man \}.
\end{equation*}
The map $\psi: \UM \to \NM$ given by $x \mapsto (x, x-\projM(x))$ is
a diffeomorphism onto its image, with inverse
$\psi(\UM) \subset \NM \to \amb$ given by $(x,n) \mapsto x+n$.

If $\man$ is a $C^j$ submanifold, then $\psi$ is a $C^{j-1}$
diffeomorphism onto its image \cite{foote1984}, and the function
$\gdtoM$ has is $C^j$ on $\UM \setminus \man$.
\end{thm}

%\subsubsection{Local reach, local feature size, and reach}

If $x \in \UM$, with $\px = \projM(x)$, then the ball
$\ballEN{x}{r}$ of radius $r= \norm{x-\px}$ is tangent to $\man$ at
$\px$, and $(x-\px) \in N_{\px}\man$. For any point $y\in \man$, the
\defn{local reach} \cite{attali2007} of $\man$ at $y$ is
\begin{equation*}
  \rch(y,\man) := \sup\{r \in \R \mid (y + ru) \in \UM\ 
 \forall u \in N_y\man \text{ with } \norm{u}=1\}.
\end{equation*}
By the tubular neighbourhood theorem,
$\rch(y,\man)$ can be equivalently defined~as
\begin{equation*}
%\label{eq:2nd.def.loc.rch}
  \rch(y,\man) = \sup\{r \in \R \mid \projM(y + ru) = y\ 
 \forall u \in N_y\man \text{ with } \norm{u}=1\}.
\end{equation*}
With this formulation, it is easy to see the following standard
observation: 

\begin{lem}
  \label{lem:empty.reach.ball}
  For any $y \in \man$,
  any open ball that is tangent to $\man$ at $y$ and with radius
  $r \leq \rch(y,\man)$, does not intersect $\man$.
\end{lem}

This property is useful for bounding the extrinsic curvatures of
$\man$ at $y$ (see, e.g.,
\cite[Lemma~3.3]{boissonnat2017convtanvar}). However, it can be
awkward to work with the local reach since it is not known to be
continuous, even when $\man$ is smooth. The \defn{local feature size}
at $y \in \man$ is the distance from $y$ to the medial axis;
equivalently, it is the supremum of the radii of balls centred at $y$
and contained in $\UM$ \cite[p.~432]{federer1959}:
\begin{equation*}
  \lfs(y) := \sup\{r \mid \ballEN{y}{r} \subset \UM \}.
\end{equation*}
Although it was introduced by Federer, the local feature size was
later rediscovered, and given its name, by Amenta and Bern
\cite{amenta1999vf}. Since the local reach at $y$ is the distance to
the medial axis measured only in directions orthogonal to $\man$ at
$y$, we have $\rch(y,\man) \geq \lfs(y)$. It is a short exercise to
show that the local feature size is 1-Lipschitz:
\begin{equation*}
  \abs{\lfs(y) - \lfs(x)} \leq \norm{x-y}, \quad \text{for all }x,y\in \man.
\end{equation*}

The \defn{reach}
of $\man$ is the infimum of the local feature size, or equivalently,
the infimum of the local reach:
\begin{equation*}
  \rchM := \distEN{\man}{\claxM} = \inf_{y\in \man} \lfs(y) =
  \inf_{y\in\man} \rch(y,\man). 
\end{equation*}
The last equality comes from the observation that since $\man$ is
compact, there is some $y^*$ for which $\lfs(y^*)=\rchM$, and we must
have $\rch(y^*,\man) = \lfs(y^*)$ because $z-y^*$ must lie in
$N_{y^*}\man$ for all $z \in \claxM \cap \cballEN{y^*}{\lfs(y^*)}$.

Observe also, that since, by the tubular neighbourhood theorem,
$\lfs(x)>0$ at any point $x\in\man$, the continuity of the local
feature size implies that $\rchM>0$ for any compact $C^2$ submanifold.

\begin{remark}[locally bounding the local reach]
  \label{rem:loc.bnd.loc.rch}
  We will often need a lower bound $\rbnd$ on $\rch(x,\man)$ in a
  neighbourhood of a point $p\in\man$, and moreover, the \emph{size}
  of this neighbourhood will depend on the size of the bound. For
  example, we will have the following awkward self-referential
  definition: $U_p = \ballEN{p}{r} \cap \man$, where $r \leq \rbnd$
  and $\rch(x,\man) \geq \rbnd$ for all $x \in U_p$.

  This can be easily resolved if we choose $\rbnd=\rchM$, the global bound
  on the local reach. However, the local reach could vary by orders of
  magnitude over the manifold, making it inefficient to use a global
  bound to govern the size of the simplices in the constructed
  simplicial approximation.
  
  The Lipschitz property of $\lfs$ makes it well-suited to bound the
  local reach in a small neighbourhood of $p\in \man$. For example,
  \begin{equation*}
    \rch(x,\man)\geq \lfs(x) \geq (1-\epsilon)\lfs(p), 
    \quad \text{for all }x \text{ with } \norm{p-x} \leq \epsilon\lfs(p).
  \end{equation*}
  Thus we can choose $r = \epsilon\lfs(p)$ and
  $\rbnd=(1-\epsilon)\lfs(p)$, and the criterion $r\leq \rbnd$ is satisfied
  provided $\epsilon \leq \frac12$.
\end{remark}

\subsection{Affine flats and angles}
\label{sec:flats.n.angles}

The angle between two vectors $u,v \in \R^N \setminus \{0\}$ 
is denoted $\angle(u,v)$ (this angle is $\leq \pi$). 
% If $K \subset\R^n$ is a linear subspace,
% then $\angle(u,K)$ is the angle between $u$ and $\proj{K}(u)$, its
% orthogonal projection into $K$ (this angle is $\leq \pi/2$). 
If $K \subseteq\R^n$ is a linear subspace, $\proj{K}(u)$ is the
orthogonal projection of into $K$. We define $\angle(u,K)$ to be
$\pi/2$ if $\proj{K}(u) = 0$, and otherwise
$\angle(u,K) = \angle(u,\proj{K}(u))$ (thus $\angle(u,K)\leq \pi/2$).
If $K$ and $L$ are two linear subspaces, 
%with $\dim K \leq \dim L$,
then
\begin{equation*}
  \angle(K,L) = \sup_{u \in K \setminus \{0\}}\angle(u,L).
\end{equation*}
The definition is only interesting when $\dim K \leq \dim L$.
Observe that if $\dim K = \dim L$, then $\angle(K,L)=\angle(L,K)$. If
$K$ and $L$ are affine flats in $\R^N$, 
%with $\dim K \leq \dim L$,
then $\angle(K,L)$ is the angle between the corresponding parallel
vector subspaces. If $\gsplxs$ is a simplex with
$\dim\gsplxs \leq \dim L$, then
$\angle(\gsplxs,L):=\angle(\aff(\gsplxs),L)$, where $\aff(\gsplxs)$ is
the affine hull of $\gsplxs$. We denote the orthogonal complement of a
linear subspace $K \subseteq \R^N$ by $K^{\perp}$. A short exercise
yields the following observations:

\begin{lem}
  \label{lem:angle.flat.complement}
  \begin{enumerate}
  \item If $K,L$ are subspaces of $\R^N$, then
    \begin{equation*}
      \angle(L^\perp,K^\perp) = \angle(K,L).
    \end{equation*}
  \item If $Q \subset \R^N$ is a subspace of codimension 1, then
    \begin{equation*}
      \angle(Q^\perp,K)= \pi/2 - \angle(K,Q).
    \end{equation*}
  \end{enumerate}
\end{lem}
%NB14p102

The following standard observation 
(cf.\ \cite[Theorem~4.8(7)]{federer1959}) follows
easily from \Lemref{lem:empty.reach.ball}. 

\begin{lem}
  \label{lem:dist.to.tanspace}
  Given any two points $x,y \in \man \subset \amb$, we have
  \begin{enumerate}
  \item $\sin\angle([x,y],T_x\man) \leq
    \frac{\norm{x-y}}{2\rch(x,\man)}$.
  \item $\distEN{y}{T_x\man} \leq \frac{\norm{x-y}^2}{2\rch(x,\man)}$.
  \end{enumerate}
\end{lem}

The following lemma is a local adaptation of results presented in
\cite{boissonnat2017convtanvar}:

\begin{lem}[tangent space variation]
  \label{lem:tan.var}
  Suppose $\ballEN{c}{r} \subset \UM$,\, 
  $x,y \in B = \ballEN{c}{r} \cap \man$ and $\rch(z,\man)\geq
  \rbnd$ for all $z$ in $B$. If $r < \rbnd$, then
  \begin{equation*}
    \sin \angle(T_x\man,T_y\man) 
    \leq
    \frac{\norm{y-x}}{\rbnd},
  \end{equation*}
and
  \begin{equation*}
    \angle(T_x\man,T_y\man) \leq \frac{\pi\norm{y-x}}{2\rbnd}.
  \end{equation*}
\end{lem}

\begin{proof}
  See \Appref{app:convex.tanvar} for a sketch of how this result
  follows from the arguments presented in
  \cite{boissonnat2017convtanvar}. 
\end{proof}

\begin{remark}
  Since the bounds in \Lemref{lem:tan.var} are vacuous if
  $\norm{y-x} \geq \rbnd$, in practice we require that either
  $c=x$ or $r<\frac12\rbnd$.
\end{remark}

We will need to bound the angle between a simplex with vertices on
$\man$ and the nearby tangent spaces. To this end we employ a result
established by Whitney \cite[p.~127]{whitney1957} in the formulation
presented in \cite[Lemma~2.1]{boissonnat2013stab1}:

\begin{lem}[Whitney's angle bound]
  \label{lem:Whitneys.angle.bnd}
  Suppose $\gsplxs$ is a $j$-simplex whose vertices all lie within a
  distance $\eta$ from a $k$-dimensional affine space, $K
  \subset \R^N$, with $k \geq j$. Then
  \begin{equation*}
    \sin \angle(\gsplxs,K) \leq
    \frac{2\eta}{tL},
  \end{equation*}
  where $t$ is the thickness of $\gsplxs$ and $L$ is the length of its
  longest edge. 
\end{lem}

\begin{lem}[simplices lie close to $\man$]
  \label{lem:splxs.close.to.man}
  Let $\gsplxs \subset \UM$ be a simplex with vertices on $\man$, 
%and with $L(\gsplxs) < \rbnd$, where 
  and $\rbnd$ a constant such that $\rch(\px,\man) \geq \rbnd$ for all
  $x \in \gsplxs$, where $\px = \projM(x)$. Then for all
  $x,y \in \gsplxs$,
  \begin{equation*}
    \distEN{y}{T_{\px}\man} < \frac{2L(\gsplxs)^2}{\rbnd},
    \quad \text{and in particular,} \quad
    \dtoM{x} < \frac{2L(\gsplxs)^2}{\rbnd}.
  \end{equation*}
\end{lem}

\begin{proof}
  Since $\px$ is the closest point on $\man$ to $x$, and every vertex
  lies on $\man$, we must have $\norm{x-\px}< L(\gsplxs)$,
%  and $\norm{x-\px}\leq \norm{p-x}$ for any vertex $p$ of $\gsplxs$,
  and therefore, for any vertex $p$ of $\gsplxs$,
  \begin{equation}
    \label{eq:dist.leq.2L}
    \norm{p-\px} \leq \norm{p-x} + \norm{x-\px} 
    % \leq 2\norm{p-x}
    < 2L(\gsplxs). 
  \end{equation}
  By \Lemref{lem:dist.to.tanspace}(2) we have $\distEN{p}{T_{\px}\man}
  < 2L^2/\rbnd$.
  This is true for all vertices $p$ of $\gsplxs$, and since the
  function $\distEN{\,\cdot\,}{T_{\px}\man}$ is affine on $\gsplxs$,
  it is also true for any  $y\in \gsplxs$. The second inequality
  follows by taking $y=x$, since
  \begin{equation*}
    \distEN{x}{T_{\px}\man} = \norm{x-\px}.
% < \frac{2L^2}{\rbnd}.
  \end{equation*}
\end{proof}

\begin{remark}
  If $p$ is a vertex of $\gsplxs$, then the constraint
  $\gsplxs \subset \UM$ of \Lemref{lem:splxs.close.to.man} is ensured
  if $L(\gsplxs) < \lfs(p)$, since $\lfs$ is the distance to the
  medial axis. In practice $\rbnd$ is defined either in terms of
  $\rchM$ or in terms of $\lfs(p)$.

  For example, since $\norm{p-\px} \leq 2L(\gsplxs)$, we have that
  $\lfs(\px) \geq \lfs(p) - 2L$. So the requirements of
  \Lemref{lem:splxs.close.to.man} are satisfied by demanding
  $L \leq \epsilon\lfs(p)$, with $\epsilon < \tfrac12$, and setting
  $\rbnd = (1-2\epsilon)\lfs(p)$. Alternatively, we can simply demand
  $L<\rchM$, and use $\rbnd=\rchM$. Of course, other variations are
  possible.
\end{remark}

\begin{lem}[simplex-tangent space angle bounds]
  \label{lem:splx.tanspace.angle.bnds}
  Suppose $\gsplxs \subset \amb$ is a simplex of dimension $\leq m$
  with vertices on $\man$. If $p$ is a vertex of $\gsplxs$, 
  % and 
  % \begin{equation*}
  %   L(\gsplxs) < t(\gsplxs)\rch(p,\man),
  % \end{equation*}
  then 
  \begin{enumerate}
  \item\hfil
    $\displaystyle
      \sin\angle(\gsplxs,T_p\man) \leq 
      \frac{L}{t\rch(p,\man)}. 
      $

  \item In addition, suppose $\gsplxs \subset \UM$, and there is a
    ball $\ballEN{c}{r} \subset \UM$ such that for any $x \in \gsplxs$,\,
    $\px = \projM(x) \in B= \ballEN{c}{r} \cap \man$ and
    $\rch(z,\man)\geq \rbnd$ for all $z\in B$. If $r < \rbnd$, then
    \begin{equation*}
      \sin\angle(\gsplxs,T_{\px}\man) 
%      \leq \frac{2L}{t\rbnd}\left(1 + \frac{L}{\rbnd}\right)
      \leq \frac{3L}{t\rbnd}.
    \end{equation*}
  \end{enumerate}
\end{lem}

\begin{proof}
  (1)\, By \Lemref{lem:dist.to.tanspace}(2), all the
  vertices of $\gsplxs$ are within a distance $\eta =
  L^2/(2\rch(p,\man))$ from $T_p\man$, and so
  \Lemref{lem:Whitneys.angle.bnd} ensures
  \begin{equation*}
    \sin\angle(\gsplxs,T_p\man) \leq \frac{2\eta}{tL} =
    \frac{L}{t\rch(p,\man)}. 
  \end{equation*}

  \smallbreak\noindent
  (2)\, By \Lemref{lem:tan.var} and 
  \eqref{eq:dist.leq.2L},
%  \Lemref{lem:splxs.close.to.man}
  \begin{equation*}
    \sin\angle(T_p\man,T_{\px}\man) \leq \frac{\norm{p-\px}}{\rbnd}
    % \leq \frac{\norm{p-x}+\norm{x-\px}}{\rbnd}
    % \leq 
    % \frac{L}{\rbnd}\left(1 + \frac{2L}{\rbnd}\right),
    \leq \frac{2L}{\rbnd},
  \end{equation*}
  and the result follows using part 1.
\end{proof}

\subsection{Distortion of orthogonal projection: $\lmap_p$ and $\phi_p$}
\label{sec:orthog.proj}

The coordinate maps $\phi_p$ and $\lmap_p$ are defined in terms of the
orthogonal projection to $T_p\man$. The size of the neighbourhoods
used to define the coordinate charts are constrained by the
requirement that these maps be embeddings, which we establish by
ensuring that they are $\xi$-distortion maps with $\xi<1$.

\begin{lem}[definition and distortion of $\phi_p$]
  \label{lem:def.distort.phi}
  Let $U_p = \ballEN{p}{r} \cap \man$, where $r=\rho \rbnd$, with
  $\rho < \frac12$, and
  $\rch(x,\man) \geq \rbnd$ for all $x \in U_p$ {\rm(see
  Remark~\ref{rem:loc.bnd.loc.rch})}. Define $\phi_p :=
\proj{T_p\man}|_{U_p}$. Then $\phi_p$ is a $\xi$-distortion map with
\begin{equation*}
  \xi = 4\rho^2.
\end{equation*}
\end{lem}

\begin{proof}
  For any distinct $x,y \in U_p$, we have, from
  \Lemref{lem:dist.to.tanspace}(1), that
  \begin{equation*}
    \sin \angle([x,y],T_x\man) \leq \rho,
  \end{equation*}
  and, from \Lemref{lem:tan.var}, that
  \begin{equation*}
    \sin \angle(T_p\man,T_x\man) \leq \rho.
  \end{equation*}
  Combining these bounds we have
  \begin{equation*}
    \sin \angle([x,y],T_p\man) \leq 2\rho.
  \end{equation*}
  Letting $\tx = \phi_p(x)$,\, $\ty = \phi_p(y)$, and $\theta =
  \angle([x,y],T_p\man)$, we find
  \begin{equation*}
    \begin{split}
      \norm{x-y} - \norm{\tx-\ty}
      &= (1-\cos\theta)\norm{x-y}\\
      &\leq (1 - \sqrt{1-4\rho^2})\norm{x-y}\\
      &\leq 4\rho^2 \norm{x-y}.
    \end{split}
  \end{equation*}
    The result follows, since $\norm{\tx-\ty}\leq \norm{x-y}$.
\end{proof}

\begin{remark}[differential of $\phi_p$]
  \label{rem:diff.phip}
  It straight forward to verify from the definitions that for any
  $x\in U_p$,
  \begin{equation*}
    d(\proj{T_p\man}|_{\man})_x = \proj{T_p\man}|_{T_x\man}.
  \end{equation*}
\end{remark}

The domain of the map $\lmap_p$, i.e., an upper bound on the allowable
size of the simplices in $\iota(\carrier{\str{p}})$, is governed by
the following bound on the metric distortion of the projection from a
simplex.

\begin{lem}[simplexwise distortion of $\lmap_p$]
  \label{lem:orth.proj.splx.distort}
  Suppose $\gsplxs \subset \amb$ is a simplex of dimension $\leq m$
  with vertices on $\man$. If $p$ is a vertex of $\gsplxs$, 
  and 
  \begin{equation*}
    L(\gsplxs) < t(\gsplxs)\rch(p,\man),
  \end{equation*}
  then the restriction of $\proj{T_p\man}$ to $\gsplxs$ is a
  $\xi$-distortion map with
  \begin{equation*}
    \xi = \left(\frac{L(\gsplxs)}{t(\gsplxs)\rch(p,\man)}\right)^2.
  \end{equation*}
\end{lem}

\begin{proof}
  Let $x,y \in \gsplxs$ and set $\x = \proj{T_p\man}(x)$,\, 
  $\y = \proj{T_p\man}(y)$. 
By \Lemref{lem:splx.tanspace.angle.bnds}(1),
  \begin{equation*}
    \sin\angle(\gsplxs,T_p\man) \leq 
    \frac{L}{t\rch(p,\man)}. 
  \end{equation*}
% By \Lemref{lem:dist.to.tanspace}, all the
%   vertices of $\gsplxs$ are within a distance $\eta =
%   L^2/(2\rch(p,\man))$ from $T_p\man$, and so
%   \Lemref{lem:Whitneys.angle.bnd} ensures
%   \begin{equation*}
%     \sin\angle(\gsplxs,T_p\man) \leq \frac{2\eta}{tL} =
%     \frac{L}{t\rch(p,\man)}. 
%   \end{equation*}
  So, putting
  $\theta=\angle([x,y],[\x,\y]) \leq \angle(\gsplxs,T_p\man)$ we find
  \begin{equation*}
    \begin{split}
      \norm{x-y} - \norm{\x-\y}
      &= (1 - \cos\theta)\norm{x-y}\\
      &\leq \left(1 -
        \sqrt{1 - \left(\frac{L}{t\rch(p,\man)}\right)^2} \right)\norm{x-y}\\
      &\leq \left(\frac{L}{t\rch(p,\man)}\right)^2\norm{x-y}.
    \end{split}
  \end{equation*}
  The result follows since $\norm{\x-\y}\leq \norm{x-y}$.
\end{proof}

Therefore, in order to use our framework, we need to ensure that for
each simplex $\gsplxs \in \acplx$, the simplex
$\iota(\gsplxs) \subset \amb$ satisfies $L < t \rbnd$, where
$\rbnd \leq \rch(p,\man)$ for each vertex $p$ of $\iota(\gsplxs)$ (we will
in fact need a stronger bound than this). Then,
in conformance with \Secref{sec:star.method}, we set
$\lmap_p = \proj{T_p\man} \circ \iota|_{\carrier{\str{p}}}$, and we
\emph{require} that it be an embedding. Although we have established
that $\proj{T_p\man}$ is an embedding on each simplex, and we have
assumed that $\iota|_{\carrier{\str{p}}}$ is an embedding, these
criteria do not imply that $\lmap_p$ is an embedding. We leave this as
a requirement of the embedding theorem (requirement (a) of
\Thmref{thm:triang.subman}), i.e., something that needs to be
established in context. For the case of the tangential Delaunay
complex, the embedding follows naturally because
$\phi_p(\carrier{\str{p}})$ is seen as a weighted Delaunay
triangulation in $T_p\man$ \cite{boissonnat2014tancplx}.

\subsection{Distortion of the closest-point projection map: $H$}
\label{sec:dist.closest.pt.proj}

% repeating the intro to Sec. 4 for Mathijs:
Recall that $H\colon \carrier{\acplx} \to \man$ is the map that we
wish to show is a homeomorphism. In our current context $H$ is based
on the closest-point projection map $\projM\colon \R^N \to \man$. As
discussed at the beginning of this section, we define
$H = \projM \circ \iota$, where $\iota\colon \carrier{\acplx}\to \R^N$
is the immersion of our simplicial complex into $\R^N$. Once $H$ is
shown to be a homeomorphism, it follows that $\iota$ is in fact an
embedding, but we don't assume this a priori. The metric on
$\carrier{\acplx}$ (i.e., the edge lengths of the Euclidean simplices)
is \emph{defined} by $\iota$, so $\iota$ itself does not contribute to
the metric distortion of $H$.

An upper bound for the metric distortion of $\projM$ was demonstrated
by \cite[Theorem~4.8(8)]{federer1959}; the proof we present here is
similar, but less general, since we require $\man$ to be a
differentiable submanifold:

\begin{lem}[upper bound for $\projM$ distortion]
  \label{lem:projM.distort.upper.bnd}
  Let $x,y \in\UM$ and $\rbnd = \min\{\rch(\px,\man),\rch(\py,\man)\}$,
where $\px=\projM(x)$, and $\py=\projM(y)$, as usual.  If
$a \geq \max\{\dtoM{x},\dtoM{y}\}$ for some $a<\rbnd$, then
  \begin{equation}
    \label{eq:projM.distort.upper.bnd}
    \norm{\py-\px} \leq \left(1-\frac{a}{\rbnd} \right)^{-1}\norm{y-x}.
  \end{equation}
\end{lem}

\begin{proof}
  Let $\tx$, $\ty$ be the orthogonal projection of $x$ and $y$ into
  the line $\ell$ generated by $\py-\px$. Then
  \begin{equation}
    \label{eq:init.upbnd}
    \norm{y-x} \geq \norm{\ty-\tx} \geq 
    \norm{\py-\px} - \norm{\py-\ty} - \norm{\px-\tx}.
  \end{equation}
  Let $Q_x$ be the hyperplane through $\px$ and orthogonal to
  $[x,\px]$. Then, using \Lemref{lem:angle.flat.complement}(2),
  we have 
  \[
    \norm{\px-\tx} = \norm{x-\px}\cos\angle([x,\px],\ell)
    = \norm{x-\px}\sin\angle(\ell,Q_x).
  \]
  Since $T_{\px}\man \subseteq Q_x$, we have
  $\angle(\ell,Q_x) \leq \angle(\ell,T_{\px}\man)$, and so by
  \Lemref{lem:dist.to.tanspace}(1),
  \begin{equation*}
    \norm{\px-\tx} \leq \frac{\norm{\py-\px}}{2\rch(\px,\man)}\norm{x-\px}
    \leq \frac{a}{2\rbnd}\norm{\py-\px}.
  \end{equation*}
  Likewise,
  \begin{equation*}
    \norm{\py-\ty} \leq \frac{a}{2\rbnd}\norm{\py-\px}.    
  \end{equation*}
  Thus \eqref{eq:init.upbnd} yields
  \begin{equation*}
    \norm{y-x} \geq \left(1 - \frac{a}{\rbnd} \right)\norm{\py-\px},
  \end{equation*}
  and hence the result.
\end{proof}

\begin{lem}[simplexwise distortion of $H$]
  \label{lem:distort.projM}
  Suppose $\gsplxs \subset \UM$ is a simplex of dimension $\leq m$
  whose vertices lie in $\man$, and there is a ball $\ballEN{c}{r}$
  such that for all $z\in\gsplxs$,\, $\pz \in B = \ballEN{c}{r} \cap
  \man$, where $\pz=\projM(z)$. Let $\rbnd$ be a lower bound on
  $\rch(\tz,\man)$ for all $\tz\in B$. If $r < \rbnd$, and
  $L(\gsplxs) < t(\gsplxs)\rbnd/3$, then
  the restriction of $\projM$ to $\gsplxs$ is a $\xi$-distortion
  map with
  \begin{equation*}
    \xi = \frac{12L^2}{t^2\rbnd^2}.
  \end{equation*}
\end{lem}

\begin{proof}
  By \Lemref{lem:splxs.close.to.man},
  \begin{equation}
    \label{eq:splx.man.dist}
    \dtoM{x} < a= \frac{2L^2}{\rbnd}  \quad\text{for any }
    x\in \gsplxs.
\end{equation}
Thus it follows from
  \Lemref{lem:projM.distort.upper.bnd} that
  \begin{equation}
    \label{eq:up.bnd.H.dist}
    \norm{\py-\px} \leq \left(1-\frac{2L^2}{\rbnd^2}
    \right)^{-1}\norm{y-x}
    \leq \left(1+ \frac{4L^2}{\rbnd^2}\right)\norm{y-x},
  \end{equation}
  for any $x,y\in\gsplxs$. 

  We now need to establish a lower bound on $\norm{\py-\px}$. Let $Q_x$
  be the hyperplane through $\px$ and orthogonal to $[x,\px]$, and let
  $\hy$ and $\hpy$ be the orthogonal projection of $y$ and $\py$ into
  $Q_x$. We have
  \begin{equation}
    \label{eq:raw.low.bnd}
    \norm{\py-\px} \geq \norm{\hpy-\px} \geq \norm{\hy-\px} - \norm{\hpy-\hy}.
  \end{equation}

  To get a lower bound on $\norm{\hy-\px} =
  \norm{y-x}\cos\angle([x,y],Q_x)$, notice that
  \begin{equation*}
    \angle([x,y],Q_x) \leq \angle(\aff(\gsplxs),Q_x)
    \leq \angle(\aff(\gsplxs),T_{\px}\man),
  \end{equation*}
  and by \Lemref{lem:splx.tanspace.angle.bnds}(2),
  \begin{equation*}
    \sin\angle(\aff(\gsplxs),T_{\px}\man) \leq \frac{3L}{t\rbnd}.
  \end{equation*}
  Thus,
  \begin{equation*}
    \cos\angle([x,y],Q_x)
    \geq
    \left(1 - \left(\frac{3L}{t\rbnd}\right)^2\right)^{\frac12}
    \geq
    \left(1 - \frac{9L^2}{t^2\rbnd^2}\right),
  \end{equation*}
  and so
  \begin{equation}
    \label{eq:raw.bnd.term.one}
    \norm{\hy-\px} \geq \left(1 - \frac{9L^2}{t^2\rbnd^2}\right)\norm{y-x}.
  \end{equation}

  To get an upper bound on the second term on the right side of
  \eqref{eq:raw.low.bnd}, let $Q_y$ be the hyperplane through $\py$
  and orthogonal to $[\py,y]$. We have
  \begin{equation*}
    \begin{split}
    \norm{\hpy-\hy}
    &= \norm{\py-y}\cos\angle([\py,y],Q_x)\\
    &= \norm{\py-y}\sin\angle(Q_x,Q_y)
    \quad \text{by \Lemref{lem:angle.flat.complement}(2)}\\
    &\leq \norm{\py-y}\sin\angle(T_{\px},T_{\py})\\
    &\leq \norm{\py-y}\frac{\norm{\py-\px}}{\rbnd}
    \quad \text{by \Lemref{lem:tan.var}}\\
    &\leq \frac{2L^2}{\rbnd^2}\norm{\py-\px}
    \quad \text{by \eqref{eq:splx.man.dist}}.
 \end{split}
  \end{equation*}
  Now, using \eqref{eq:up.bnd.H.dist}, we have
  \begin{equation}
    \label{eq:raw.bnd.term.two}
    \norm{\hpy-\hy}
    \leq \frac{2L^2}{\rbnd^2}\left(1+
      \frac{4L^2}{\rbnd^2}\right)\norm{y-x}
    < \frac{3L^2}{\rbnd^2} \norm{y-x},
  \end{equation}
  since the hypothesis $3L/\rbnd < t$ implies $4L^2/\rbnd^2 < \frac12$.

  Finally, plugging \eqref{eq:raw.bnd.term.one} and
  \eqref{eq:raw.bnd.term.two} back into \eqref{eq:raw.low.bnd}, we get
  \begin{equation*}
    \norm{\py-\px} \geq \left(1-\frac{12L^2}{t^2\rbnd^2} \right)\norm{y-x}.
  \end{equation*}
  Comparing this lower bound with \eqref{eq:up.bnd.H.dist}, we arrive
  at the stated value for the metric distortion $\xi$.
\end{proof}

\subsection{Triangulation criteria for submanifolds}

In order to employ \Thmref{thm:metric.triang} we first ensure that we
meet the compatible atlases criteria
(\Defref{def:compatible.atlases}). In \Lemref{lem:def.distort.phi} we defined
\begin{equation*}
  U_p = \ballEN{p}{r} \cap \man,
  \quad \text{where }
  r = \rho\rbnd < \tfrac12\rbnd. 
\end{equation*}
We need to ensure that $H(\carrier{\str{p}}) \subseteq U_p$. In our
context, this means that we require $\projM(\iota(\carrier{\str{p}}))
\subset U_p$, and 
%\eqref{eq:dist.leq.2L} 
\Lemref{lem:splxs.close.to.man}
ensures that if $L(\gsplxs)
\leq \Lbnd$ for all $\gsplxs \in \iota(\str{p})$, then it is
sufficient to choose $r=\Lbnd + 2\Lbnd^2/\rbnd$, or
\begin{equation}
  \label{eq:def.rho}
  \rho = \frac{\Lbnd}{\rbnd}\left(1+\frac{2\Lbnd}{\rbnd}\right).
\end{equation}
Our bound on $\Lbnd$ itself will be much smaller than
$\rbnd$, which in turn will be expressed in terms of $\lfs(p)$ or
$\rchM$, so we will have $r<\lfs(p)$, and thus $\ballEN{p}{r} \subset
\UM$, ensuring also that each $\gsplxs \in \iota(\str{p})$ also lies
in $\UM$.

We need to establish the metric distortion of
$F_p = \phi_p \circ H \circ \lmap_p^{-1}$ restricted to any
$m$-simplex in $\lmap_p(\iota(\str{p}))$, and ensure that it meets the
distortion control criterion of \Thmref{thm:metric.triang}.

Anticipating the bound we will need to meet the distortion-control
criterion of \Thmref{thm:metric.triang}, we impose the constraint
\begin{equation}
  \label{eq:prelim.Lbnd}
\chnk \leq \frac{1}{16^2},
\end{equation}
where $\tbnd$ is a lower bound on the thickness: $t(\gsplxs)\geq
\tbnd$ for all $\gsplxs \in \iota(\str{p})$. We remark that $\Lbnd$
and $\rbnd$ may be considered to be \emph{local constants}, i.e., they
may depend on the vertex $p \in \pts$, however $\tbnd$ and the ratio
$\Lbnd/\rbnd$ will be global constants.

Using \eqref{eq:prelim.Lbnd} together with Lemmas
\ref{lem:orth.proj.splx.distort} and \ref{lem:inv.comp.distort}(1), we
can bound the metric distortion of $\lmap_p^{-1}$ as
\begin{equation*}
  \xi_1 = \frac{L^2}{t^2\rbnd^2} 
  \left(1-\frac{L^2}{t^2\rbnd^2} \right)^{-1}
  \leq \left(\frac{16^2}{16^2-1}\right)\chnk.
%  \leq \left(\frac{400}{399}\right)\chnk.
\end{equation*}

\Lemref{lem:distort.projM} gives us the distortion of $H$:
\begin{equation*}
  \xi_2 = 12\chnk.
\end{equation*}

For $\phi_p$, using \Lemref{lem:def.distort.phi} and
\eqref{eq:def.rho} we get the distortion bound
\begin{equation*}
  \xi_3 = 4\rho^2
  = \frac{4\Lbnd^2}{\rbnd^2}\left(1 + \frac{2\Lbnd}{\rbnd} \right)^2
  \leq \frac{9}{2}  \chnk.
\end{equation*}

\Lemref{lem:inv.comp.distort}(2) says that the distortion of $F_p$ is
no more than
\begin{equation*}
  \xi = \xi_1 + \xi_2 + \xi_3 + \xi_1\xi_2 + \xi_1\xi_3 + \xi_2\xi_3
  + \xi_1\xi_2\xi_3.
\end{equation*}
Using \eqref{eq:prelim.Lbnd} we find that $F_p$ is a $\xi$-distortion
map with
\begin{equation}
  \label{eq:defn.xi}
  \xi = \chink{19}.
\end{equation}
%Note: 
% The calculations for this section are done in calc.py. They are
% self-referential, in that the bound in \eqref{eq:prelim.Lbnd} cannot
% be stricter than the final bound found in \eqref{eq:Lbnd.squared},
% but if it is too loose, the final bound will get bigger.
Observe that \eqref{eq:prelim.Lbnd} implies that all the maps involved
have distortion less than~1.

Now we need to ensure that this bound meets the distortion-bound
requirement for \Thmref{thm:metric.triang}. We have chosen to use here
the properties of the Euclidean simplices in the ambient space $\amb$,
but in \Thmref{thm:metric.triang} we are considering simplices in the
local coordinate space; for us these are the projected simplices, e.g.,
$\hat{\gsplxs} = \lmap_p(\gsplxs) = \proj{T_p\man}(\gsplxs)$.  Using
\Lemref{lem:orth.proj.splx.distort}, the distortion properties of the
affine map $\proj{T_p\man}$ imply
\begin{equation*}
  a(\hat{\gsplxs}) \geq \left(1 - \chnk \right)a(\gsplxs),
  \quad
  L(\hat{\gsplxs}) \leq L(\gsplxs),
\end{equation*}
and therefore
\begin{equation*}
  t(\hat{\gsplxs}) \geq \left(1 - \chnk \right)t(\gsplxs).
\end{equation*}
We can thus set 
\begin{equation*}
  \htbnd = \left(1 - \chnk \right)\tbnd,
  \quad
  \hLbnd = \Lbnd,
  \quad
  \hsbnd = \left(1 - \chnk \right)\sbnd,
\end{equation*}
where $\sbnd$ is a lower bound for the diameters of the simplices
in $\iota(\str{p})$.

Then, in order to meet the distortion control criterion of
\Thmref{thm:metric.triang}, we require
\begin{equation}
  \label{eq:Lbnd.needed}
  \chink{19} < \frac{\hsbnd\htbnd^2}{12\hLbnd}
  = \left(1 - \chnk \right)^3 \frac{\sbnd\tbnd^2}{12\Lbnd}.
\end{equation}

It is convenient to define $\sepbnd = \sbnd/\Lbnd$. Then, observing
that
\begin{equation*}
\left(1 - \chnk \right)^3 \geq \left(1 - 3\chnk \right),
\end{equation*}
we see that \eqref{eq:Lbnd.needed} is satisfied if
\begin{equation}
  \label{eq:Lbnd.squared}
  \Lbnd^2 \leq \frac{\sepbnd\tbnd^4\rbnd^2}{16^2}.
\end{equation}

Now we consider $\rbnd$. For any $x\in U_p$, the Lipschitz
continuity of $\lfs$ ensures that $\lfs(x) > \lfs(p) - \rho\rbnd$,
where $\rho$ is given by \eqref{eq:def.rho}. Our
constraint \eqref{eq:Lbnd.squared} on $\Lbnd$ implies
\begin{equation}
  \label{eq:bnd.r}
  \rho\rbnd \leq \frac{9}{8}\Lbnd \leq \frac{9}{2^7}\rbnd.
\end{equation}
Thus we need 
%$\rbnd \leq \lfs(p) - \frac{9}{2^7}\rbnd$, 
$\rbnd \leq \lfs(p) - \frac{9}{128}\rbnd$, 
which is
satisfied by
\begin{equation}
  \label{eq:lfs.def.rbnd}
  \rbnd = \frac{128}{137}\lfs(p).
\end{equation}
Of course, we can also choose $\rbnd=\rchM$ independent of
$p$. Plugging these values back into \eqref{eq:Lbnd.squared} gives us
two alternatives for the bound on $\Lbnd$:
\begin{equation}
  \label{eq:two.Lbnds}
  \Lbnd \leq \frac{\sepbnd^{\frac12}\tbnd^2\lfs(p)}{18}
  \quad\text{or}\quad
  \Lbnd \leq \frac{\sepbnd^{\frac12}\tbnd^2\rchM}{16}.
\end{equation}

We now only need to establish that $F_p$ is simplexwise positive
to arrive at our triangulation theorem for submanifolds.

\begin{lem}
  $F_p$ is simplexwise positive on
  $\proj{T_p\man}(\iota(\carrier{\str{p}}))$. 
\end{lem}

\begin{proof}
  Recall that $F_p = \phi_p \circ H \circ \lmap_p^{-1}$, and observe
  that the restriction of $F_p$ to an $m$-simplex
  $\hat\gsplxs \in \proj{T_p\man}(\iota(\str{p}))$ is differentiable.
  The bound \eqref{eq:Lbnd.squared} together with
  \eqref{eq:defn.xi} implies $F_p$ is a $\xi$-distortion map with
  $\xi<1$, so
  \Lemref{lem:distort.bnd.differential} ensures that for any
  $m$-simplex $\gsplxs \in \proj{T_p\man}(\iota(\str{p}))$, the
  differential of $F_p$ does not vanish on $\gsplxs$. 

We choose an orientation on $U_p$; thus we have an
orientation on each tangent space $T_x\man$,\, $x\in U_p$. The
projection map $\proj{T_p\man}|_{T_x\man} =
d(\proj{T_p\man}|_{\man})_x = d(\phi_p)_x$
(see Remark~\ref{rem:diff.phip}) is then orientation preserving, because
of continuity: it is certainly true when $x=p$, and
Lemmas \ref{lem:def.distort.phi} 
and \ref{lem:distort.bnd.differential}
imply that 
the differential $d(\phi_p)$ is
nondegenerate on $U_p$. Thus $\phi_p$ is simplexwise positive.

For an $m$-simplex $\gsplxs \in \iota(\str{p})$ we \emph{define} the
orientation such that $\proj{T_p\man}|_{\gsplxs}$ is
positive. Equivalently, we define the orientation on
$\iota(\carrier{\str{p}})$ to be such that $\lmap_p$ is simplexwise
positive. This is not problematic, since we require that $\lmap_p$ be
an embedding. Thus $\lmap_p$ is simplexwise positive by our definitions.

It remains to show that $\projM|_{\gsplxs}$ is positive for each
$m$-simplex $\gsplxs \in \iota(\str{p})$. First observe that for any
$x\in \R^N$, the kernel of $d(\projM)_x$ is the subspace of $T_x\R^N$
corresponding to $N_{\px}\man$ (under the canonical identification
$T_x\R^N \cong \R^N \cong T_{\px}\R^N$, where $\px=\projM(x)$). Also,
observe that if $x\in \man$, then
$d(\projM|_{T_x\man})_x = \id_{T_x\man}$. This is easily seen by
observing that $d(\projM|_{T_x\man})_x = d(\projM)_x|_{T_x\man}$ and
using curves on $\man$ to apply the definition of the differential. 

Thus at the central vertex $p \in \gsplxs$, we can express
$d(H|_{\gsplxs})_p$ as
\begin{equation*}
  d(\projM|_{\gsplxs})_p = \id_{T_p\man} \circ \proj{T_p\man}|_{\gsplxs}.
\end{equation*}
Since we have established that $\proj{T_p\man}|_{\gsplxs}$ is
positive, it follows that $d(\projM|_{\gsplxs})_p$ is positive, and since
$d(\projM|_{\gsplxs})$ is nondegenerate on $\gsplxs$ (because $d(F_p)$
is), it follows that 
$H|_{\iota(\carrier{\str{p}})}$ is simplexwise positive.

Thus $F_p$ is orientation preserving on each simplex of $\str{\p}$.
\end{proof}

\begin{thm}[triangulation for submanifolds]
  \label{thm:triang.subman}
  Let $\man \subset \R^N$ be a compact $C^2$ manifold, and
  $\pts \subset \man$ a finite set of points such that for each
  connected component $\cmpman$ of $\man$,\,
  $\cmpman\cap \pts \neq \emptyset$. Suppose that $\acplx$ is a
  simplicial complex whose vertices, $\acplx^0$, are identified with
  $\pts$, by a bijection $\acplx^0\to\pts$ such that the resulting
  piecewise linear map $\iota\colon \carrier{\acplx} \to \R^N$ is an
  immersion, i.e., $\iota|_{\carrier{\str{p}}}$ is an embedding for
  each vertex $p$.

  If:
  \begin{enumerate}\def\theenumi{\alph{enumi}}
  \item For each vertex $p\in \pts$, the projection
    $\proj{T_p\man}|_{\iota(\carrier{\str{p}})}$ is an embedding and
    $p$ lies in the interior of $\proj{T_p\man}(\iota(\carrier{\str{p}}))$.
  \item There are constants $0<\tbnd\leq1$,\, $0<\sepbnd \leq 1$, and
    $\epsilon_0>0$ such that for each simplex
    $\gsplxs \in \iota(\acplx)$, and each vertex $p \in \gsplxs$,
    \begin{equation*}
      t(\gsplxs)\geq \tbnd,
      \quad \sepbnd\epsilon_0\lfs(p) \leq L(\gsplxs) \leq
      \epsilon_0\lfs(p),
      \quad \epsilon_0 \leq \frac{\sepbnd^{\frac12}\tbnd^2}{18}.
    \end{equation*}
  \item For any vertices $p,q \in \pts$, if
\[
q \in U_p = \ballEN{p}{r} \cap \man, 
\quad\text{where } r= \frac{\lfs(p)}{15},
%$r=9\lfs(p)/137$,
\]
    then $\proj{T_p\man}(q) \in \proj{T_p\man}(\iota(\str{p}))$ if and
    only if $q$ is a vertex of $\str{p}$.
    \end{enumerate}
    Then:
    \begin{enumerate}
    \item $\iota$ is an embedding, so the complex $\acplx$ may be
      identified with $\iota(\acplx)$.
    \item The closest-point projection map
      $\projM|_{\carrier{\acplx}}$ is a homeomorphism
      $\carrier{\acplx} \to \man$.
    \item For any $x \in \gsplxs \in \acplx$,
      \begin{equation*}
        \dtoM{x} = \norm{\px-x} \leq \tfrac73\epsilon_0^2\lfs(\px),
        \quad
        \sin\angle(\gsplxs, T_{\px}) \leq \frac{13\epsilon_0}{4\tbnd},
      \end{equation*}
      where $\px = \projM(x)$.
    \end{enumerate}
\end{thm}

\begin{proof}
  Observe that the local embedding condition (a) implies that $\acplx$
  is a compact $m$-manifold without boundary. Condition (b) is a
  reformulation of the first inequality of
  \eqref{eq:two.Lbnds}.
%  \eqref{eq:lfs.def.rbnd}.
  Condition (c) is the
  vertex sanity condition of \Thmref{thm:metric.triang};
  % , where the bound on $r$ comes from
  using
  \eqref{eq:bnd.r} and \eqref{eq:lfs.def.rbnd} we get $r =
  9\lfs(p)/137 < \lfs(p)/15$.

  Thus, from our argument above, the criteria of
  \Thmref{thm:metric.triang} are satisfied, and
  $H = \projM \circ \iota$ is a homeomorphism. It follows that $\iota$
  is injective, and since $\carrier{\acplx}$ is compact, $\iota$ must
  be an embedding. The second consequence, that
  $\projM|_{\carrier{\acplx}}$ is a homeomorphism, is now immediate.

  For the third consequence, notice that, as argued to obtain
  \eqref{eq:lfs.def.rbnd}, the Lipschitz continuity of $\lfs$ imples that
  \begin{equation*}
    \lfs(p) \leq \frac{137}{128}\lfs(\px).
  \end{equation*}
  So, using \Lemref{lem:splxs.close.to.man} and
  \eqref{eq:lfs.def.rbnd}, we have
  \begin{equation*}
    \dtoM{x} \leq \frac{2L^2}{\rbnd}
    \leq
    \frac{2\epsilon_0^2\lfs(p)^2}{\rbnd}
    \leq
    2\left(\frac{137}{128}\right)^2\epsilon_0^2\lfs(\px)
    \leq \tfrac73\epsilon_0^2\lfs(\px).
  \end{equation*}

  The second inequality follows from
  \Lemref{lem:splx.tanspace.angle.bnds}(2) and
  \eqref{eq:lfs.def.rbnd}.
\end{proof}

If we use the global bound $\rbnd=\rchM$, to bound the size of the
simplices, then the third consequence of \Thmref{thm:triang.subman}
can be tightened. In this context we obtain the following variation of
\Thmref{thm:triang.subman}, which is a corollary in the sense that it
follows from essentially the same proof, even though it does not
follow from the statement of \Thmref{thm:triang.subman}.

\begin{cor}
  \label{cor:rch.subman.triang}
  If the conditions {\rm(b)} and {\rm(c)} in \Thmref{thm:triang.subman} are
  replaced by
  \begin{enumerate}\addtocounter{enumi}{1}\def\theenumi{\alph{enumi}$'$}
    \def\labelenumi{\textup{(\theenumi)}}
  \item There are constants $0<\tbnd\leq1$,\, $0<\sepbnd \leq 1$, and
    $\epsilon_0>0$ such that for each simplex
    $\gsplxs \in \iota(\acplx)$, and each vertex $p \in \gsplxs$,
    \begin{equation*}
      t(\gsplxs)\geq \tbnd,
      \quad \sepbnd\epsilon_0\rchM \leq L(\gsplxs) \leq
      \epsilon_0\rchM,
      \quad \epsilon_0 \leq \frac{\sepbnd^{\frac12}\tbnd^2}{16}.
    \end{equation*}
  \item For any vertices $p,q \in \pts$, if
    $q \in U_p = \ballEN{p}{r} \cap \man$, where
    % $r=9\rchM/128$,
    $r=\rchM/14$,
    then $\proj{T_p\man}(q) \in \proj{T_p\man}(\iota(\str{p}))$ if and
    only if $q$ is a vertex of $\str{p}$.
  \end{enumerate}
  then the conclusions of \Thmref{thm:triang.subman} hold, and
  consequence (3) can be tightened to:
  \begin{enumerate}\addtocounter{enumi}{2}\def\theenumi{\arabic{enumi}$'$}
    \def\labelenumi{\textup{(\theenumi)}}
    \item For any $x \in \gsplxs \in \acplx$,
      \begin{equation*}
        \dtoM{x} = \norm{\px-x} \leq 2\epsilon_0^2\rchM,
        \quad
        \sin\angle(\gsplxs, T_{\px}) \leq \frac{3\epsilon_0}{\tbnd},
      \end{equation*}
      where $\px = \projM(x)$.
    \end{enumerate}
\end{cor}

\begin{proof}
  This follows from the proof of \Thmref{thm:triang.subman}, using
  $\rbnd=\rchM$ instead of \eqref{eq:lfs.def.rbnd}, e.g., from
  \eqref{eq:Lbnd.squared} we obtain the second alternative in
  \eqref{eq:two.Lbnds}, and from \eqref{eq:bnd.r} we get $r=9\rchM/128
  < \rchM/14$.
\end{proof}

%%% Local Variables:
%%% mode: latex
%%% TeX-master: "homeo"
%%% End:

%% file: degree.tex
% -*- LaTeX -*-
% degree.tex
% 20170318
%
\section{Elementary degree theory}
\label{sec:degree.theory}

We recall here some basic ideas in degree theory. Our primary
motivation is to facilitate the statement and proof of
\Lemref{lem:splx.pos.embed}, recovering what we need of a result of
Whitney \cite[Appendix~II Lemma 15a]{whitney1957}, without using
differentiability assumptions.

We are interested in the degree of continuous maps
$F\colon \cl{\Omega} \to \R^m$, where $\Omega \subset \R^m$ is an open,
bounded, and nonempty domain in $\R^m$, as can be found Chapter~IV,
Sections 1 and 2 of \cite{outerelo2009degree}, for example.
As with most modern treatments of degree theory, Outerelo and Ruiz start by
defining the degree for a regular value of a differentiable map, where
the idea is transparent. For $y \in \R^m \setminus F(\bdry{\Omega})$,
the degree is defined as
\begin{equation*}
  \deg(F,\Omega,y) := \sum_{x\in F^{-1}(y)}\sgn_x(F),
\end{equation*}
where $\sgn_x(F)$ denotes the sign ($\pm1$) of the Jacobian determinant
(the determinant of the differential of $F$) at $x$. Thus the degree
at $y$ counts the number of points in the preimage, accounting for the
local orientation of the map.

It is then shown that the degree is locally constant on the (open) set
of regular values, and after showing that it is also invariant under
homotopies that avoid conflicts between $y$ and the image of the
boundary, the degree is defined for an arbitrary point in
$\R^m \setminus F(\bdry{\Omega})$, and it is constant on each
connected component of $\R^m \setminus F(\bdry{\Omega})$.

Then the definition of degree is extended to continuous maps $F\colon
\cl{\Omega} \to \R^m$ by means of the Weierstrass approximation
theorem (here $\onorm{\,\cdot\,}_{\infty}$ denotes the supremum norm): 

\begin{lem}[{{\cite[Proposition and Definition IV.2.1]{outerelo2009degree}}}]
\label{lem:deg.diff.approx}
Let $F\colon \cl{\Omega}   \to \R^m$ be a continuous map, and let $y
\in \R^m \setminus F(\bdry{\Omega})$. Then there exists a smooth map
$G: \cl{\Omega} \to \R^m$ such that $\onorm{F-G}_{\infty} <
\distEm{y}{F(\bdry{\Omega})}$. For all such $G$, the degree
$\deg(G,\Omega,y)$ is defined ($y\in \R^m\setminus G(\bdry{\Omega})$)
and is the same, and we define the degree of $F$ by
\begin{equation*}
  \deg(F,\Omega,y)=\deg(G,\Omega,y).
\end{equation*}
Furthermore, $G$ can be chosen such that $y$ is a regular value of
$G|_{\Omega}$, and then
\begin{equation*}
  \deg(F,\Omega,a) = \sum_{x\in G^{-1}(y)} \sgn_x(G).
\end{equation*}
\end{lem}

The locally constant nature of the degree is the main property we wish
to exploit:
\begin{lem}[{{\cite[Proposition  IV.2.3]{outerelo2009degree}}}]
  \label{lem:deg.locally.const}
  Let $F\colon \cl{\Omega} \to \R^m$ be a continuous map. Then the
  degree $y \mapsto \deg(F,\Omega,y)$ is constant on every connected
  component of $\R^m \setminus F(\bdry{\Omega})$. 
\end{lem}

Notice that if $F$ is a topological embedding, then
$\deg(F,\Omega,y) = \pm 1$ for any point $y \in F(\Omega)$.  We will
also have occasion to use the following:

\begin{lem}[{{\cite[Corollary IV.2.5(3)]{outerelo2009degree}}}]
  \label{lem:deg.additive}
  Given a continuous mapping $F \colon \cl{\Omega} \to \R^m$,  two
  disjoint open subsets $\Omega_1, \Omega_2 \subset \Omega$, and a
  point $y \not\in F(\cl{\Omega}\setminus (\Omega_1 \cap \Omega_2))$,
  \begin{equation*}
    \deg(F,\Omega,y) = \deg(F,\Omega_1,y) + \deg(F,\Omega_2, y).
  \end{equation*}
\end{lem}

Since no connectedness assumptions are made on the open sets in
question, a straightforward inductive argument allows us to strengthen
the statement of \Lemref{lem:deg.additive}:

\begin{lem}
  \label{lem:deg.multi.add}
  Suppose $\Omega_1, \ldots, \Omega_n$ are mutually disjoint open subsets of
  the open domain
  $\Omega \subset \R^m$, and $F\colon \cl{\Omega} \to \R^m$ is a
  continuous map. If $y \not\in F(\cl{\Omega} \setminus
  \bigcup_{i=1}^n\Omega_i)$, then
  \begin{equation*}
    \deg(F,\Omega,y) = \sum_{i=1}^n \deg(F,\Omega_i,y).
  \end{equation*}
\end{lem}

\subsection{Orientation and cogent maps}
\label{sec:cogent}

A simplex is \defn{oriented} by choosing an orientation for its affine
hull, or equivalently, by ordering its vertices; any even permutation
of this order describes the same orientation. An $m$-simplex in $\R^m$
has a natural orientation induced from the canonical orientation of
$\R^m$ defined by the standard basis. Our convention is that $\gsplxs
\subset \R^m$ is positively oriented if its vertices $v_i$,\, $0\leq
i\leq m$, are ordered such that the basis $\{v_1-v_0, \ldots,
v_m-v_0\}$ defines the same orientation as the canonical basis of
$\R^m$. 

In the case that concerns us, where $\ccplx$ is a finite pure
$m$-complex piecewise linearly embedded in $\R^m$ (so that we can
naturally view $\carrier{\ccplx}\subset \R^m$), we assume that the
$m$-simplices carry the canonical orientation inherited from the
ambient space.

\begin{de}[orientation preserving map]
  \label{def:orientation.preserving}
If $\gsplxs \subset \R^m$ is an $m$-simplex, we say that a continuous
topological embedding $F\colon \gsplxs \to \R^m$ is \defn{orientation
  preserving}, or \defn{positive}, if $\deg(F,\rint{\gsplxs},y)=1$, for
any point $y\in \intr{F(\gsplxs)}$, otherwise $F$ is \defn{orientation
  reversing}.
\end{de}

\begin{de}[cogent maps]
  \label{def:cogent.maps}
Suppose $\ccplx$ is a pure $m$-complex. We call a map $F\colon
\carrier{\ccplx} \to \R^m$ \defn{cogent} with respect to $\ccplx$ if 
it is continuous and
its restriction to each simplex is an embedding. 
\end{de}

For cogent maps, the points in
$\R^m \setminus F(\carrier{\ccplx^{m-1}})$ are analogous to the
regular values of a differentiable map. If $F$ is a cogent map and
$x\in \rint{\gsplxs}$, where $\gsplxs$ is an $m$-simplex in $\ccplx$,
we define
\begin{equation}
  \label{eq:def.sgn.cogent}
  \sgn_x(F) := \deg(F,\rint{\gsplxs}, F(x)).
\end{equation}
This makes the analogy with the degree of a differentiable map
transparent: 

\begin{lem}[degree of cogent maps]
  \label{lem:deg.cogent.map}
  Suppose $\ccplx$ is a finite pure $m$-complex embedded in $\R^m$. Let
  $\Omega = \carrier{\ccplx} \setminus \bdry{\carrier{\ccplx}}$.
  If a continuous map $F\colon \carrier{\ccplx} \to \R^m$ is cogent
  with respect to $\ccplx$, then for any $y \in \R^m \setminus
  F(\carrier{\ccplx^{m-1}})$ {\rm (recall that $\bdry{\carrier{\ccplx}}
  \subseteq \carrier{\ccplx^{m-1}}$ \cite[Lemmas~3.6,
  3.7]{boissonnat2013stab1})},
  \begin{equation*}
    \deg(F,\Omega,y) = \sum_{x \in F^{-1}(y)} \sgn_x(F).
  \end{equation*}
\end{lem}

A very minor modification to the proof of
\Lemref{lem:deg.diff.approx} would allow us to sharpen the bound on
$\onorm{F-G}_{\infty}$ in our context, so that
\begin{equation*}
\onorm{F-G}_{\infty} < 
\min_{\gsplxs\cap F^{-1}(y)\neq \emptyset}
\distEm{y}{F(\bdry{\gsplxs})},
\end{equation*}
and \Lemref{lem:deg.cogent.map} follows immediately from the same proof.
But rather than digging into the proof of that lemma, we can avoid
getting our hands dirty and just exploit 
\Lemref{lem:deg.multi.add}.

\begin{proof}[of \Lemref{lem:deg.cogent.map}]
  The preimage of $y$ is a finite set of points:
  $F^{-1}(y) = \{x_i\}$, $i\in \{1,\ldots,n\}$. For each $i$, let
  $\gsplxs_i$ be the $m$-simplex that contains $x_i$ in its interior,
  and set $\Omega_i = \rint{\gsplxs_i}$. Then the statement follows
  from \Lemref{lem:deg.multi.add} and the definition
  \eqref{eq:def.sgn.cogent} of $\sgn_{x_i}(F)$.
\end{proof}

\begin{de}[simplexwise positive]
  \label{def:simplexwise.positive}
  A map is \defn{simplexwise positive} if it is cogent and its
  restriction to any $m$-simplex is orientation preserving.
\end{de}

Lemmas
\ref{lem:deg.locally.const}
and
\ref{lem:deg.cogent.map}
yield:

\begin{lem}
  \label{lem:splxwise.pos.loc.cnst}
  If $\ccplx$ is a finite pure $m$-complex embedded in $\R^m$, and
  $F \colon \carrier{\ccplx} \to \R^m$ is simplexwise positive, then
  for any connected open subset $W$ of
  $\R^m \setminus F(\bdry{\carrier{\ccplx}})$, any two points of $W$
  not in $F(\carrier{\ccplx^{m-1}})$ are covered the same number of
  times (i.e., have the same number of points in their preimage under
  $F$).
\end{lem}

\begin{remark}
  Whitney's Lemma \cite[Appendix~II Lemma 15]{whitney1957} is a
  combination of Lemmas \ref{lem:splxwise.pos.loc.cnst} and
  \ref{lem:splx.pos.embed}, but applies in more generality. However,
  he assumes that the restriction of $F$ to each $m$-simplex is a
  smooth map.  The generality can be fully recovered without invoking
  this differentiability assumption.

  The definition of the degree of a map $\cl{\Omega} \to \R^m$ can be
  naturally extended to the case where $\cl{\Omega}$ is an oriented
  abstract manifold with boundary. Using this, the assumption that
  $\ccplx$ is embedded in $\R^m$ can be dropped. Also, the assumption
  that $\ccplx$ be finite can be relaxed, at least if we assume that
  $F$ is proper (so that the number of points in the preimage of a
  point is finite --- Whitney seems to assume this).

  Whitney also only assumed that $\ccplx$ was a \defn{pseudomanifold}
  with boundary: a pure $m$-complex such that any $(m{-}1)$-simplex
  is a face of either 2 or 1 $m$-simplices (those that are the face of
  only 1 $m$-simplex define the boundary complex).

  Brouwer's original exposition of degree theory \cite{brouwer1912}
  used simplicial approximations, and piecewise linear maps, rather
  than differentiable maps as the foundation. Even though the
  simplicial aspect is attractive and natural in our setting, there
  would be no economy in using this approach for our
  purposes. However, Brouwer's exposition was based on the notion of
  pseudomanifolds \cite[p.\,23]{outerelo2009degree}, so that approach
  would also allow us to recover the pseudomanifold aspect of
  Whitney's lemma.
\end{remark}

%%% Local Variables:
%%% mode: latex
%%% TeX-master: "homeo"
%%% End:

%% file: tanvar.tex
\section{On tangent space variation}
\label{app:convex.tanvar}

The purpose of this appendix is to sketch the demonstration of
\Lemref{lem:tan.var}, which is adapted from arguments presented in
%\cite{boissonnat2017convtanvar}. 
\cite{boissonnat2017convtanvarsup}.
% \rd{Actually, the argument I had was adapted from
%   \cite{boissonnat2017convtanvarsup}, so I will just clean that up for
%   now. It would be nicer to only reference the primary document
%   \cite{boissonnat2017convtanvar}.  It occurred to me that this entire
%   appendix could be made unnecessary, and the statements of many of
%   the results could be considerably simplified, if we could show that
%   the following is true:
%
%   \textit{If $B \subset \UM$ is a closed Euclidean ball, then
%     $\rch(B\cap \man)\geq \rchM$.}
%
% Do you think that is true? Is there a short concise proof?
%
% Actually, that was just a guess; I have no compelling reason to
% believe it; just a hope that there is a nicer characterization than
% that given in \Lemref{lem:convexity}. Mathijs's new result implies
% that if $A\subseteq \man$ is geodesically convex, then
% $\rch(A)\geq \rchM$. Also, if $\ballEN{c}{r}$ has $r<\rchM$, then
% $\ballEN{c}{r}\cap \man$ is geodesically convex. My question is: what
% is the right \emph{local bound} on $r$ for this to be true?  }
The argument relies on this convexity result:

\begin{lem}[convexity]
  \label{lem:convexity}
  Suppose $\ballEN{c}{r} \subset \UM$ is such that
  $\rch(x,\man)\geq \rbnd$ for all
  $x \in B = \ballEN{c}{r} \cap \man$. If $r < \rbnd$, then $B$ is
  geodesically convex in the sense that for any $x,y\in B$, any
  minimizing geodesic between $x$ and $y$ is contained in $B$.  
\end{lem}

\begin{proof}
This follows from the same argument that produced
\cite[Theorem~3.6]{boissonnat2017convtanvarsup}, but using $\rbnd$
instead of $\rchM$. Here is an overview of the adjustments that must
be done to the argument:

\cite[Corollary~2.3]{boissonnat2017convtanvarsup} holds with $\rchM$
replaced by $\rch(p,\man)$.

\cite[Lemma~3.1]{boissonnat2017convtanvarsup} holds with $\rch(M)$
replaced with $\rbnd$, where $\rbnd$ is a lower bound on
$\rch(\gamma(t),\man)$, for all relevant $t$. 

\cite[Lemma~3.2]{boissonnat2017convtanvarsup} also holds more
generally. In fact the essential argument has nothing to do with reach
or manifolds, it says this:

If $\alpha \colon I \to \R^N$ is parameterized by arc length, and has
curvature bounded by $1/R$, i.e., $\norm{\alpha''(t)} \leq 1/R$ for all
$t\in I$, then for any $a,b\in I$,
\begin{equation}
  \label{eq:curv.tan.angle.bnd}
  \angle(\alpha'(a),\alpha'(b)) \leq \frac{\len(\alpha([a,b]))}{R}.
\end{equation}
This comes directly from using the bound on the curvature to give a
bound on the length of the ``indicatrix of tangents'', i.e., the curve
traced out on the sphere by the unit tangent vectors.

This allows us to bound the distance between the endpoints of
$\alpha(I)$ for sufficiently small $I$. Let $I=[0,\ell]$, and
$\alpha(0)=a$ and $\alpha(\ell)=b$. Then, letting $v=\alpha'(\ell/2)$
and integrating $\ip{b-a}{v} = \int_0^\ell \ip{\alpha'(s)}{v}\, ds$ in
two parts, we get
\begin{equation}
  \label{eq:bnd.endpt.dist}
  \norm{b-a} \geq 2R\sin\left(\frac{\ell}{2R}\right),
\quad\text{assuming } \ell \leq \pi R.
\end{equation}
(Actually, it seems the argument is okay as long as $\ell \leq 2\pi
R$, but after $\ell$ becomes larger than $\pi R$ the bound becomes
smaller, finally vanishing when we've come full circle.)

\cite[Lemma~3.3]{boissonnat2017convtanvarsup}: Again this is really
just an argument about curvature controlled curves. The stated bound
holds if $\gamma$ is any space curve, and $\rchM$ is replaced with
$R$, where $\norm{\gamma''}\leq R$, and
$\len(\gamma)<\pi R$. 

\cite[Lemma~3.4]{boissonnat2017convtanvarsup}: 
% \rd{This is the lemma that
% was missing the first time I read the document}
We can replace this statement with: If $p$ and $q$ are connected by a
minimizing geodesic $\gamma$, and $\norm{p-q}<2\rbnd$, where $\rbnd$
is a lower bound on the local reach along $\gamma$, then
$\distM{p}{q}<\pi\rbnd$. 

The lens-shaped region described in
\cite[Corollary~3.5]{boissonnat2017convtanvarsup} is of course
constructed using the short arc of a circle of radius $\rbnd$, where
$\rbnd$ is a lower bound on the local reach along the geodesic.

Now the convexity argument for
\cite[Corollary~3.5]{boissonnat2017convtanvarsup} goes through when we
replace the ball of radius less than reach with the ball
$\ballEN{c}{r} \subset \UM$ such that $\rch(x,\man)\geq \rbnd$ for all
$x \in B = \ballEN{c}{r} \cap \man$, and $r < \rbnd$.
Indeed, note that the condition $\ballEN{c}{r} \subset \UM$ ensures that
$\ballEN{c}{r}$ can only intersect a single connected component of
$\man$ (since the medial axis separates topological components).
\end{proof}

Using the kind of argument that leads to
\Eqnref{eq:curv.tan.angle.bnd} we find: If $\rch(x,\man) \leq \rbnd$
for all $x$ lying on a minimizing geodesic between $a$ and $b$ on
$\man$, then
\begin{equation}
  \label{eq:geod.bnd.tan.angle}
  \angle(T_a\man,T_b\man) \leq \frac{\distM{a}{b}}{\rbnd}.
\end{equation}
This is found by using the curvature bound provided by $\rbnd$ to bound
the angle between any vector $u \in T_a\man$ and $v \in T_b\man$,
obtained by parallel transport of $u$ along the geodesic to $b$. This
bound does not require a bound on the distance between $a$ and $b$,
but it becomes vacuous if $\distM{a}{b} \geq \pi \rbnd/2$.

Using \eqref{eq:bnd.endpt.dist}, we get
\begin{equation}
  \label{eq:bnd.half.tan.angle}
  \sin \left(\tfrac12 \angle(T_a\man,T_b\man) \right)
\leq
\frac{\norm{a-b}}{2\rbnd},
\end{equation}
when $\rbnd$ is a lower bound on the local reach along the geodesic,
and the length of a minimizing geodesic doesn't exceed $\pi\rbnd$. As discussed
in the ``proof'' of \Lemref{lem:convexity}, the argument of 
\cite[Lemma~3.4]{boissonnat2017convtanvarsup} implies that if
$a$ and $b$ are connected by a
minimizing geodesic $\gamma$, and $\norm{a-b}<2\rbnd$, where $\rbnd$
is a lower bound on the local reach along $\gamma$, then
$\distM{a}{b}<\pi\rbnd$.

\begin{lem}
  \label{lem:raw.tan.var}
  Suppose $a,b \in \man$ are connected by a minimizing geodesic
  $\gamma$, and $\rch(x,\man)\leq R$ for all $x$ along $\gamma$. Then
  \begin{equation*}
    \sin \angle(T_a\man,T_b\man)
    \leq
    \frac{\norm{a-b}}{R}.
  \end{equation*}
  If $\norm{a-b} \leq 2R$, then
  \begin{equation*}
    \angle(T_a\man,T_b\man) \leq \frac{\pi\norm{a-b}}{2R}.
  \end{equation*}
\end{lem}

\begin{proof}
  The first inequality follows from \eqref{eq:bnd.half.tan.angle} and
  the observation (from the angle sum formula) that $\sin(2\theta)
  \leq 2\sin\theta$.

  The second claim follows from the observation above that $\norm{a-b}
  \leq 2\rbnd$ implies $\distM{a}{b} \leq \pi \rbnd$, and so
  \eqref{eq:geod.bnd.tan.angle} shows that $\theta = \frac12
  \angle(T_a\man,T_b\man) \leq \pi/2$. We then use the observation that
  in this case $\frac{2\theta}{\pi} \leq \sin\theta$.
\end{proof}

Combining Lemmas \ref{lem:convexity} and \ref{lem:raw.tan.var}, we
obtain \Lemref{lem:tan.var}.

%%% Local Variables:
%%% mode: latex
%%% TeX-master: "homeo"
%%% End:

%% file: strong_diff.tex
\section{Exploiting strong differential bounds}
\label{sec:exploit.strong.diff.bnds}

The triangulation criteria of \Thmref{thm:metric.triang} presented in
\Secref{sec:homeo.criteria} are based on the triangulation
demonstration presented in
\cite[Proposition~16]{dyer2015riemsplx}. The main motivation for
presenting the new argument in \Secref{sec:homeo.criteria} is that the
methods of \cite{dyer2015riemsplx} require an intricate analysis of
the differential of the map $F_p$, which makes the application of the
result, considerably more difficult than meeting the purely metric
criteria of \Thmref{thm:metric.triang}.

The motivation for reviewing the previous method here is that the
demonstration of the triangulation result
\cite[Proposition~16]{dyer2015riemsplx} was incorrect, and in fact the
statement of the proposition does not ensure the injectivity of the
map $H$. We correct the problem here by employing the vertex sanity
assumption (\Defref{def:vertex.sanity}) introduced in
\Secref{sec:injectivity}, and provide an erratum in
\Secref{sec:bary.map.erratum}. Although the criteria for this method
are more difficult to establish, once they are established, a stronger
result is obtained, as mentioned in \Remref{rem:bnd.linear.in.t}, so
these results may still be of interest.

For the method of \Thmref{thm:metric.triang}, if $F_p$ is
differentiable on each $m$-simplex, as is the case in the setting of
\Secref{sec:subman}, then \Lemref{lem:differential.bnd.distort} says
that to show that $F_p$ is a $\xi$-distortion map, it is sufficient to
show that for any vector $w$ tangent to a point $u$ in the domain of
$F_p$,
\begin{equation}
  \label{eq:weak.diff.bnd}
  (1-\xi)\norm{w} \leq \norm{d(F_p)_u w} \leq (1+\xi)\norm{w}.
\end{equation}
In the analogous result, \cite[Proposition~16]{dyer2015riemsplx}, a
stronger bound on the differential is demanded: we require that
\begin{equation}
  \label{eq:strong.diff.bnd}
  \onorm{d(F_p)_u -\id} \leq \xi
\end{equation}
for all $u$ in the domain. The bound \eqref{eq:strong.diff.bnd} is
strictly stronger than \eqref{eq:weak.diff.bnd}. It is not difficult
to establish that \eqref{eq:strong.diff.bnd} implies that $F_p$ is a
$\xi$-distortion map \cite[Lemma~11]{dyer2015riemsplx} on each
simplex. However, whereas \eqref{eq:weak.diff.bnd} only constrains how
much $dF_p$ can change the magnitude of a vector,
\eqref{eq:strong.diff.bnd} also constrains how much the direction can
change. For this kind of bound on the differential, there is no need
to exploit the trilateration lemma
(\Lemref{lem:distort.trilateration}):

\begin{lem}[{{\cite[Lemma~12]{dyer2015riemsplx}}}]
  \label{lem:strong.almost.id}
  Suppose $\conset \subseteq \R^m$ is a convex set and
  $F\colon \conset \to \R^m$ is a smooth map with a fixed point
  $p \in \conset$. If
  \begin{equation*}
    \onorm{dF_x - \Id} \leq \xi \quad \text{for all } x \in \conset,
  \end{equation*}
  then
  \begin{equation*}
    \norm{F(x)-x} \leq \xi\norm{x-p} \quad \text{for all }x \in \conset.
  \end{equation*}
\end{lem}

The added control obtained by the strong bound
\eqref{eq:strong.diff.bnd} on the differential enables us to ensure
that $F_p$ embeds the boundary of $\str{p}$. This in turn implies that
$F_p$ embeds all of $\str{p}$. This follows from the following
corrollary to Whitney's lemma \cite[Lemma~AII.15a]{whitney1957} (which
is demonstrated as Lemmas \ref{lem:splxwise.pos.loc.cnst} and
\ref{lem:splx.pos.embed} in this work). We include the proof here
since the proof in \cite{dyer2015riemsplx} erroneously states that
Whitney's proof implies that a simplexwise positive map is a local
homeomorphism; this is not true, but such a map is an open map.

\begin{de}[smooth on $\ccplx$]
  \label{def:smooth.on.C}
  Given a simplicial complex $\ccplx$, we say that a map
  $F\colon \carrier{\ccplx} \to \rem$ is \defn{smooth on $\ccplx$} if
  for each $\gsplxs \in \ccplx$ the restriction $F \big|_{\gsplxs}$ is
  smooth.
\end{de}

\begin{lem}[{{\cite[Lemma~13]{dyer2015riemsplx}}}]
  \label{lem:whitney.cplx.embed}
  Let $\ccplx$ be a (finite) simplicial complex embedded in $\rem$
  such that $\intr{\carrier{\ccplx}}$ is nonempty and connected, and
  $\bdry{\carrier{\ccplx}}$ is a compact, connected
  $(m{-}1)$-manifold. Suppose $F \colon \carrier{\ccplx} \to \rem$ is
  smooth on $\ccplx$ and simplexwise positive. If the restriction of
  $F$ to $\bdry{\carrier{\ccplx}}$ is an embedding, then $F$ is a
  topological embedding.
\end{lem}

\begin{proof}
  The assumptions on $\intr{\carrier{\ccplx}}$ and
  $\bdry{\carrier{\ccplx}}$ imply that $\ccplx$ is a pure $m$-complex,
  and that each $(m{-}1)$-simplex is either a boundary simplex, or the
  face of exactly two $m$-simplices.
% Indeed: assumption  bdry is manifold means any splx of dim < m is a face of 
% an (m-1)-simplex. An (m-1)-simplex cannot be the face of more than 2
% m-simplices, because C is embedded. So we just need to show that an
% (m-1)-simplex must be the face of an m-simplex. In the original
% statement (in riem-tri), an isolated (m-1)-simplex wasn't ruled
% out. But I have now added the assumption that the boundary is
% connected, so that takes care of that. 
%
% Now, if an (m-1)-simplex that is not the face of an $m$-simplex,
% shares a face with an m-simplex, then the bdry cannot be manifold. 
% indeed there must be a shared face of dim m-2, otherwise the manifold
% criterion is immediately violated. But an m-2 splx that is a face of
% an m-simplex must be the face of at least two distinct m-1 faces of
% m-simplices, so this m-2 splx would have to have at least 3 m-1
% cofaces, contradicting the manifold criterion.
  (This is a nontrivial exercise, and requires the demand that
  $\bdry{\carrier{\ccplx}}$ be connected, which was absent in the
  original statement of the lemma.)

  By the same argument as in the proof of \Lemref{lem:splx.pos.embed},
  $F(\intr{\carrier{\ccplx}})$ is open.  By the Jordan-Brouwer
  separation theorem~\cite[\S IV.7]{outerelo2009degree},
  $\R^n \setminus F(\bdry{\carrier{\ccplx}})$ consists of two open
  components, one of which is bounded. Since $F(\carrier{\ccplx})$ is
  compact, \Lemref{lem:splxwise.pos.loc.cnst} implies that
  $F(\intr{\carrier{\ccplx}})$ must coincide with the bounded
  component, and in particular
  $F(\intr{\carrier{\ccplx}}) \cap F(\bdry{\carrier{\ccplx}}) =
  \emptyset$, so $F(\intr{\carrier{\ccplx}})$ is a single connected
  component.

  We need to show that $F$ is injective. First we observe that the set
  of points in $F(\intr{\carrier{\ccplx}})$ that have exactly one
  point in the preimage is nonempty. It suffices to look in a
  neigbhourhood of a point $y \in F(\bdry{\carrier{\ccplx}})$. Choose
  $y = F(x)$, where $x$ is in the relative interior of
  $\gsplxs^{m-1} \subset \bdry{\carrier{\ccplx}}$. Then there is a
  neighbourhood $V$ of $y$ such that $V$ does not intersect the image
  of any other simplex of dimension less than or equal to $n-1$. Let
  $\gsplxs^m$ be the unique $m$-simplex that has $\splxs^{m-1}$ as a
  face. Then $F^{-1}(V \cap F(\carrier{\ccplx}) \subset \gsplxs^m$,
  and it follows that every point in $V \cap \intr{\carrier{\ccplx}}$
  has a unique point in its image.

  Now the injectivity of $F$ follows from
  \Lemref{lem:splx.pos.embed}. 
\end{proof}

We then obtain bounds on the differential of $F_p$ sufficient to
ensure that it embeds $\str{p}$:

\begin{lem}[{{\cite[Lemma~14]{dyer2015riemsplx}}}] %[Embedding a star]
  \label{lem:embed.star}
  Suppose $\ccplx = \str{\p}$ is a $\thickbnd$-thick, pure $m$-complex
  embedded in $\rem$ such that all of the $m$-simplices are incident
  to a single vertex, $\p$, and $\p \in \intr{\carrier{\ccplx}}$ {\rm(i.e.,
  $\str{\p}$ is a full star)}. If $F \colon \carrier{\ccplx} \to \rem$
  is smooth on $\ccplx$, and satisfies
  \begin{equation}
    \label{eq:embed.str.norm.bnd}
    \onorm{dF - \Id}  < m\thickbnd
  \end{equation}
  on each $m$-simplex of $\ccplx$, then $F$ is an embedding.
\end{lem}

If the requirements of \Lemref{lem:embed.star} are met, then $H$ is a
local homemorphism, and we are left with ensuring that it is
injective, in order to guarantee that it is a homeomorphism. To this
end we employ the vertex sanity assumption
(\Defref{def:vertex.sanity}), and we arrive at the following
triangulation result, which can replace the flawed
\cite[Proposition~16]{dyer2015riemsplx}: 

\begin{prop}[triangulation with differential control]
  \label{prop:strong.bnd.triang}
  Suppose $\acplx$ is a compact $m$-manifold complex (without
  boundary), with vertex set $\pts$, and $\man$ is an $m$-manifold. A
  map $H \colon \carrier{\acplx} \to \man$ is a homeomorphism if the
  following criteria are satisfied:
  \begin{enumerate}
  \item \textbf{\textup{compatible atlases}}\, There are compatible
    atlases
\[
\{(\tccplx_p,\lmap_p)\}_{p\in\pts}, \quad \tccplx_p \subset \acplx, 
\quad and \quad
\{(U_p,\phi_p)\}_{p\in\pts}, \quad U_p \subset \man,
\]
    for $H$, with
    $\tccplx_p = \str{p}$ for each $p \in \pts$, the vertex set of~$\acplx$
    {\rm(\Defref{def:compatible.atlases})}.
  \item \textup{\bf simplex quality}\, For each $p \in \pts$, every
    simplex $\gsplxs \in \strp = \lmap_p(\str{p})$ satisfies
    $s_0 \leq L(\gsplxs) \leq L_0$ and $t(\gsplxs) \geq t_0$ {\rm
      (Notation~\ref{not:splx.qual})}.
  \item \textup{\bf distortion control}\, For each $p\in \pts$, the map
    \[
      F_p = \phi_p \circ H \circ \lmap_p^{-1}
      \colon
      \carrier{\strp} \to \rem,
    \]
    is smooth on $\strp$, and simplexwise positive, and for any
    $m$-simplex $\gsplxs \in \strp$ and any $u\in \gsplxs$, we have
    \[
      \onorm{d(F_p)_u - \id} < \xi = \frac{s_0 t_0}{2 L_0}
    \]
    {\rm(Definitions \ref{def:smooth.on.C} and
      \ref{def:simplexwise.positive})}.
  \item \textup{\bf vertex sanity}\, For all vertices $p,q \in \pts$,
    if $\phi_p \circ H(q) \in \carrier{\strp}$, then $q$ is a vertex
    of~$\str{p}$.
  \end{enumerate}
\end{prop}

\begin{proof}
  Since the requirements of \Lemref{lem:embed.star} are met, we only
  need to show that $H$ is injective. The argument the same as for
  \Lemref{lem:injectivity}, with minor modifications.

Towards a contradiction, suppose that $H(q) \in H(\gsplxs)$ and that
$q$ is not a vertex of the $m$-simplex $\gsplxs$. This means there is
some $x \in \gsplxs$ such that $H(x) = H(q)$. Let $p$ be a vertex
of~$\gsplxs$.  The vertex sanity hypothesis
(\Defref{def:vertex.sanity}) implies that $\phi_p \circ H(q)$ must be
either outside of $\carrier{\strp}$, or belong to its boundary.  Thus
Lemmas \ref{lem:dist.to.star.bdry} and \ref{lem:strong.almost.id}, and
the bound on $\xi$ imply that the barycentric coordinate of $x$ with
respect to $p$ must be smaller than $\frac{1}{m+1}$: Let
$\x = \lmap_p(x)$, and $\hat{\gsplxs}=\lmap_p(\gsplxs)$.
\Lemref{lem:strong.almost.id} says that
  \begin{equation*}
    \abs{F_p(\x)-\x} < \xi L_0 = \frac{s_0 t_0}{2} \leq \frac{a_0}{2m},
  \end{equation*}
  where $a_0$ is a lower bound on the altitudes of $\p$, as in
  \Lemref{lem:dist.to.star.bdry}. Since
  $F_p(\x)=\phi_p \circ H(x)=\phi_p\circ H(q)$ is at least as far away
  from $\x$ as $\bdry{\strp}$, \Lemref{lem:dist.to.star.bdry} implies
  that the barycentric coordinate of $\x \in \hat{\gsplxs}$ with
  respect to $\p$ must be strictly less than
  $\frac{1}{2m} \leq \frac{1}{m+1}$. Since $\lmap_p$ preserves
  barycentric coordinates, and the argument works for any
  vertex $p$ of $\gsplxs$, we conclude that all the barycentric
  coordinates of $x$ in $\gsplxs$ are strictly less than
  $\frac{1}{m+1}$. We have reached a contradiction with the fact that
  the barycentric coordinates of $x$ must sum to 1.
\end{proof}

\begin{remark}
  \label{rem:bnd.linear.in.t}
  Notice that the constraint on $\xi$ in
  \Propref{prop:strong.bnd.triang} is only linear in $t_0$, whereas it
  is quadratic in $t_0$ in \Thmref{thm:metric.triang}.
\end{remark}

%\section{The barycentric coordinate map}
\subsection{Erratum for \textit{Riemannian simplices and triangulations}}
\label{sec:bary.map.erratum}

% from macros_riem.tex
\newcommand{\ren}{\reel^n}
\newcommand{\seccurv}{K}
\newcommand{\curvabsbnd}{\Lambda}
\newcommand{\scale}{h} 
\newcommand{\tanspace}[2]{T_{#1}{#2}} % e.g. T_xS
\newcommand{\gthickness}{t}
\newcommand{\thickness}[1]{\gthickness(#1)} 
\let\oldpts\pts
\renewcommand{\pts}{S}
\newcommand{\gdistA}{\gdistG{\acplx}} % intrinsic metric on complex
\newcommand{\distA}[2]{\distG{\acplx}{#1}{#2}}

This subsection is an erratum for \cite{dyer2015riemsplx}, and we will
employ here the notation and conventions of that paper, which differ
slightly from those in the rest of the current document. In
particular, the dimension of the Riemannian manifold we are
triangulating is $n$, and $\splxs$ denotes an abstract simplex, i.e.,
a set of vertices, usually on the manifold, $\man$. When these
vertices are lifted via the inverse of the exponential map to
$T_p\man$, the resulting vertex set is denoted $\splxs(p)$.  A
``filled in'' Euclidean simplex is denoted $\splxsE$.  The vertex set
of $\acplx$ is denoted by $\pts$ instead of $\oldpts$.

The proof of the generic triangulation criteria, Proposition~16 in
\cite{dyer2015riemsplx}, is flawed; the criteria presented do not
guarantee that the map $H$ is injective. This problem infects the
results in \cite{dyer2015riemsplx} which rely on Proposition~16:
Theorem~2, Proposition~26, and Theorem~3. 

These results all hold true without any further modifications if the
following hypothesis is added to Proposition~16, and Theorem~2 (the
other results simply require that the hypotheses of Theorem~2 are
met): 

\begin{hyp}[simple injectivity assumption]
  \label{hyp:simple.inject.assumpt}
  If $q$ is a vertex of $\acplx$ and $q\in H(\splxsE)$, then $q$ is a
vertex of $\gsplxs$.
\end{hyp}

The injectivity of the map $H$ follows trivially from
\Hypref{hyp:simple.inject.assumpt}, since it has been established that
$H$ is a covering map. However, this assumption is not easy to verify,
at least not in some applications of interest, e.g.,
\cite{boissonnat2017manmesh}. We provide here corrected statements of
the affected results, obtained by replacing
\cite[Proposition~16]{dyer2015riemsplx} with
\Propref{prop:strong.bnd.triang}, which uses the vertex sanity
assumption, \Defref{def:vertex.sanity}.
Unfortunately, the bound required on the differential 
in {prop:strong.bnd.triang}(3)
is qualitatively
different from that imposed in
\cite[Proposition~16]{dyer2015riemsplx}. In particular, we need to
impose a lower bound $s_0$ on the diameters of the simplices. It is
convenient to define $\mu_0 = s_0/L_0$. 

In the proof of \cite[Theorem~2]{dyer2015riemsplx}, the Proposition~16
is employed at the bottom of page~23, just before the statement of the
theorem. Reworking that short calculation, and incorporating the
vertex sanity hypothesis, we obtain:

\begin{thm}[{{\cite[Theorem~2]{dyer2015riemsplx} corrected}}]
  \label{thm:riem.triangulation.corrected}
  Suppose $\man$ is a compact $n$-dimensional Riemannian manifold with
  sectional curvatures $\seccurv$ bounded by
  $\abs{\seccurv} \leq \curvabsbnd$, and $\acplx$ is an abstract
  simplicial complex with finite vertex set $\pts \subset \man$.
Define  quality parameters $\thickbnd > 0$, $0< \mu_0 \leq 1$, and let
    \begin{equation}
      \label{eq:new.bnd}
      \scale = \min \left\{ \frac{\injradM}{4},
      \frac{\sqrt{\mu_0}\thickbnd}{6 \sqrt{\curvabsbnd}} \right\}.
    \end{equation}
Suppose
  \begin{enumerate}
  \item For every $p \in \pts$, the vertices of $\str{p}$ are
    contained in $\ballM{p}{h}$, and the balls
    $\{\ballM{p}{h}\}_{p \in \pts}$ cover $\man$.
  \item For every $p \in \pts$, the restriction of the inverse of the
    exponential map $\exp_p^{-1}$ to the vertices of
    $\str{p} \subset \acplx$ defines a piecewise linear embedding of
    $\carrier{\str{p}}$ into $\tanspace{p}{\man}$, realising $\str{p}$
    as a full star, $\widehat{\str{p}}$, such that every simplex
    $\splxs(p)$ has thickness $\thickness{\splxs(p)} \geq \thickbnd$
    and diameter $\mu_0L_0 \leq L(\splxs(p)) \leq L_0$.
  \item For all vertices $p,q \in \pts$, if
    $(\exp_p|_{\ballM{p}{h}})^{-1}(q) \in
    \carrier{\widehat{\str{p}}}$, then $q$ is a vertex of~$\str{p}$.
  \end{enumerate}
  Then $\acplx$ triangulates $\man$, and the triangulation is given by
  the barycentric coordinate map on each simplex.
\end{thm}

For \cite[Proposition~26]{dyer2015riemsplx}, the affected argument is
in the paragraph preceding the statement of the proposition. We
replace that paragraph with (where it is understood that references to
Theorem~2 are to the corrected version,
\Thmref{thm:riem.triangulation.corrected}).
The result is that the new (stronger) bounds on $\scale$ 
 are always sufficient to ensure a piecewise flat metric on $\acplx$:

\begin{quote}
  Thus in order to guarantee that the $\ell_{ij}$ describe a
  non-degenerate Euclidean simplex, we require that
  $ \curvabsbnd \scale^2 = \eta\thickbnd^2/2$,
  for some non-negative $\eta < 1$.

  Under the conditions of Theorem~2 we may have
  $\scale^2 \curvabsbnd = \frac{\mu_0 \thickbnd^2}{36}$, which gives
  us $\eta = \frac{\mu_0}{18} \leq \frac{1}{18} < 1$. Thus the
  requirements of Theorem~2 are sufficient to ensure the existence of
  a piecewise flat metric on $\acplx$, and we obtain:
\end{quote}

\begin{prop}[{{\cite[Proposition~26]{dyer2015riemsplx}}} corrected]
  \label{prop:pwflat.metric.exists.corrected}
  If the requirements of Theorem~2, are satisfied then the geodesic
  distances between the endpoints of the edges in $\acplx$ define a
  piecewise flat metric on $\acplx$ such that each simplex
  $\splxs \in \acplx$ satisfies
  \begin{equation*}
    \thickness{\splxs} > \frac{3}{4\sqrt{n}}\thickbnd.
  \end{equation*}
\end{prop}

Likewise, for \cite[Theorem~3]{dyer2015riemsplx}, it is sufficient to
impose the new bounds introduced in
\Thmref{thm:riem.triangulation.corrected} to obtain the same metric
distortion bound:

\begin{thm}[{{\cite[Theorem~3]{dyer2015riemsplx}}} corrected]
  \label{thm:metric.distortion.corrected}
  If the requirements of Theorem~2 are satisfied,
  then $\acplx$ is naturally equipped with a piecewise flat metric
  $\gdistA$ defined by assigning to each edge the geodesic distance in
  $\man$ between its endpoints.

  With this metric on $\acplx$, if
  $H \colon \carrier{\acplx} \to \man$ is the triangulation defined by
  the barycentric coordinate map, then the metric distortion induced
  by $H$ is quantified as
  \begin{equation*}
    \abs{ \distM{H(x)}{H(y)} - \distA{x}{y} } \leq \frac{50
      \curvabsbnd \scale^2}{\thickbnd^2} \distA{x}{y}, 
  \end{equation*}
  for all $x,y \in \carrier{\acplx}$.
\end{thm}

Finally, we remark that \cite[Theorem~3]{boissonnat2017manmesh}
employed \cite[Theorem~3]{dyer2015riemsplx}, but although some of the
discussion leading up to the statement of the result should be
modified to account for the new bound imposed by
\Thmref{thm:metric.distortion.corrected}, the actual result
\cite[Theorem~3]{boissonnat2017manmesh} stands as stated, since the
bound on the sampling density required there is manifestly sufficient
to accommodate \eqref{eq:new.bnd}, and the construction via local
Delaunay triangulations in the coordinate domains automatically
ensures that the vertex sanity criterion is satisfied.

%%% Local Variables:
%%% mode: latex
%%% TeX-master: "homeo"
%%% End: